\documentclass[reqno]{amsart}
\usepackage[all]{xy}

\SelectTips{cm}{}

\usepackage{amsbsy,mathrsfs,latexsym,amssymb,mathbbol,stmaryrd}
\usepackage[mathcal]{euscript}

\theoremstyle{plain}
\newtheorem{theorem}{Theorem}[section]
\newtheorem{proposition}[theorem]{Proposition}
\newtheorem{corollary}[theorem]{Corollary}
\newtheorem{lemma}[theorem]{Lemma}

\theoremstyle{definition}
\newtheorem{definition}[theorem]{Definition}
\newtheorem{example}[theorem]{Example}
\newtheorem{examples}[theorem]{Examples}
\newtheorem{assumption}[theorem]{Assumption}

\newtheorem{remark}[theorem]{Remark}

\newtheorem{notation}[theorem]{Notation}

\newtheorem{construction}[theorem]{Construction}

 \usepackage[colorlinks,pagebackref=true,citecolor=blue]{hyperref}

\numberwithin{equation}{section}

\renewcommand{\phi}{\varphi}

\def\op{{\mathit{op}}}
\def\co{{\mathit{co}}}

\def\id{{\mathit{id}}}

\def\wh#1{\widehat{#1}}
\def\wt#1{\widetilde{#1}}
\def\ol#1{\overline{#1}}

\def\cotensor{\pitchfork}

\def\Lan#1#2{#2/#1}
\def\Ran#1#2{#2{\setminus}#1}

\def\dd{{\sslash}}

\def\Lim#1#2{\{#1,#2\}}
\def\Colim#1#2{#1\star #2}
\def\colim{\mathop{\mathrm{colim}}}
\def\coins{\mathop{\mathrm{coins}}}

\def\Ord{{\mathrm{Ord}}}

\def\H{{\mathcal{H}}}

\def\E{{\mathcal{E}}}
\def\M{{\mathcal{M}}}

\def\LInj#1{{\mathsf{LInj}}(#1)}
\def\wLInj#1{{\mathsf{LInj}}_{{\mathit{w}}}(#1)}
\def\RInj#1{{\mathsf{RInj}}(#1)}

\def\TT{{\mathbb{T}}}

\def\kat#1{{\mathscr{#1}}}

\def\Set{{\mathsf{Set}}}
\def\Pos{{\mathsf{Pos}}}
\def\Top{{\mathsf{Top}}}

\def\C{\kat{C}}
\def\D{\kat{D}}
\def\X{\kat{X}}
\def\Y{\kat{Y}}

\renewcommand{\to}{\longrightarrow}
\def\into{\hookrightarrow}

\begin{document}
\title[Kan injectivity in order-enriched categories]
      {Kan injectivity in order-enriched categories}
\author{Ji\v{r}\'{\i} Ad\'{a}mek}
\address{Institute of Theoretical Computer Science, Technical University of Braunschweig,
        Germany}
\email{adamek@iti.cs.tu-bs.de}
\author{Lurdes Sousa}
\address{Polytechnic Institute of Viseu \& Centre for Mathematics of the University of Coimbra, Portugal}
\email{sousa@estv.ipv.pt}
\thanks{The second author acknowledges the support of  the Centre for Mathematics of the University of Coimbra 
       (funded by the  program COMPETE and by  the Funda\c c\~ao para a Ci\^encia e a Tecnologia, 
        under the project PEst-C/MAT/UI0324/2013).}
\author{Ji\v{r}\'{\i} Velebil}
\address{Department of Mathematics, Faculty of Electrical Engineering, Czech Technical University
         in Prague, Czech Republic}
\email{velebil@math.feld.cvut.cz}
\thanks{The third author acknowledges the support
        of the grant  No.~P202/11/1632
        of the Czech Science Foundation.} 
\date{7 November, 2013}

\begin{abstract}
Continuous lattices were characterised by 
Mart\'{\i}n Escard\'{o} as precisely the objects that are Kan-injective
w.r.t. a certain class of morphisms. We study Kan-injectivity in general
categories enriched in posets. An example: $\omega$-CPO's are precisely
the posets that are Kan-injective w.r.t. the embeddings 
$\omega\into\omega+1$ and $0\into 1$. 

For every class $\H$ of morphisms we study the subcategory of all objects
Kan-injective w.r.t. $\H$ and all morphisms preserving Kan-extensions.
For categories such as $\Top_0$ and $\Pos$ we prove that whenever $\H$
is a set of morphisms, the above subcategory is monadic, and the monad it creates
is a Kock-Z\"{o}berlein monad. However, this does not generalise to
proper classes: we present a class of continuous mappings in $\Top_0$
for which Kan-injectivity does not yield a monadic category.
\end{abstract}

\maketitle

\begin{flushright}
\footnotesize
Dedicated to the memory of Daniel M.~Kan (1927--2013)    
\end{flushright}

\section{Introduction}
\label{sec:intro}

Dana Scott's result characterising continuous lattices 
as precisely the injective topological $T_0$-spaces,
see~\cite{S72}, was one of the milestones of domain
theory. This was later refined by Alan Day~\cite{day} 
who characterised continuous lattices as the algebras for the open filter
monad on the category $\Top_0$ of topological $T_0$-spaces
and by Mart\'{\i}n Escard\'{o}~\cite{E1}
who used the fact that the category $\Top_0$ of topological
$T_0$-spaces is naturally enriched in the category of posets
(shortly: order-enriched).

In every order-enriched category one can define 
the {\em left Kan extension\/} $\Lan{h}{f}$ of a morphism
$f:A\to X$ along a morphism $h:A\to A'$ 
\begin{equation}
\label{eq:1.1}
\vcenter{
\xymatrixcolsep{1pc}
\xymatrix{
A
\ar[0,2]^-{h}
\ar[1,1]_{f}
&
\ar@{}[1,0]|(.4){\leq}
&
A'
\ar[1,-1]^{\Lan{h}{f}}
\\
&
X
&
}  
}
\end{equation}
as the smallest morphism from $A'$ to $X$ with
$f\leq (\Lan{h}{f})\cdot h$. An object $X$ is called
{\em left Kan-injective\/} w.r.t. $h$ iff for every
morphism $f$ the left Kan extension $\Lan{h}{f}$ exists
and fulfills $f=(\Lan{h}{f})\cdot h$. 
Mart\'{\i}n Escard\'{o} proved that in $\Top_0$ the left
Kan-injective spaces w.r.t. all subspace inclusions are
precisely the continuous lattices endowed with the Scott topology. 
And w.r.t. all dense subspace inclusions they are precisely the continuous
Scott domains (again with the Scott topology), see~\cite{E1}.

Recently, Margarida Carvalho and Lurdes Sousa~\cite{CS} 
extended the concept
of left Kan-injectivity to morphisms: a morphism is left-Kan
injective w.r.t. $h$ if it preserves left Kan extensions
along $h$.

We thus obtain, for every class
$\H$ of morphisms in an order-enriched category $\X$, a
(not full, in general) subcategory
$$
\LInj{\H}
$$
of all objects and all morphisms that are left Kan-injective
w.r.t. every member of $\H$. 

\begin{example}
For $\H$ = subspace embeddings in $\Top_0$, $\LInj{\H}$ is the
category of continuous lattices (endowed with the Scott topology)
and meet-preserving continuous maps.  
\end{example}

\begin{example}
In the category $\Pos$ of posets take $\H$ to consist of the two
embeddings $\omega\into\omega+1$ and $\emptyset\into 1$. Then 
$\LInj{\H}$ is the category of $\omega$-CPOS's, i.e., posets
with a least element and joins of $\omega$-chains, and $\omega$-continuous
strict functions.
\end{example}

We are going to prove that whenever the subcategory $\LInj{\H}$
is reflective, i.e., its embedding into $\X$ has a left
adjoint, then the monad $\TT=(T,\eta,\mu)$ on $\X$ that this adjunction defines
is a Kock-Z\"{o}berlein monad, i.e., the inequality $T\eta\leq \eta T$ holds.
And $\LInj{\H}$ is the Eilenberg-Moore category $\X^\TT$. 
Our main result is that in a wide class of order-enriched categories,
called {\em locally ranked categories\/} (they include $\Top_0$ and $\Pos$),
every class $\H$ of morphisms, such that all members of $\H$ but a set are
order-epimorphisms, defines a reflective subcategory $\LInj{\H}$.
However, this does not hold for general classes $\H$: we present a class $\H$
of continuous functions in $\Top_0$ whose subcategory $\LInj{\H}$ fails to 
be reflective.

We also study weak left Kan-injectivity: this means that for every $f$
a left Kan extension $\Lan{h}{f}$ exists but in~\eqref{eq:1.1} equality
is not required. We prove that, in a certain sense, this concept can always
be substituted by the above (stronger) one.

\section{Left Kan-injectivity}
\label{sec:lkan-inj}

Throughout the paper we work with
\begin{enumerate}
\item
{\em order-enriched categories\/} $\X$, i.e., all homsets
$\X(X,X')$ are partially ordered, and composition is monotone
(in both variables)  
\end{enumerate}
and
\begin{enumerate}
\setcounter{enumi}{1}
\item
{\em locally monotone functors\/} $F:\X\to\Y$, i..e, 
the derived functions from $\X(X,X')$ to $\Y(FX,FX')$
are all monotone. 
\end{enumerate}

\begin{notation}
Given morphisms 
$$
\xymatrixcolsep{1pc}
\xymatrix{
A
\ar[0,2]^-{h}
\ar[1,1]_{f}
&
&
A'
\\
&
X
&
}
$$
we denote by $\Lan{h}{f}:A'\to X$ the {\em left Kan extension\/}
of $f$ along $h$. That is, we have $f\leq (\Lan{h}{f})\cdot h$
and for all $g:A'\to X$
\begin{equation}
\label{eq:2.1}
\vcenter{
\xymatrixcolsep{1pc}
\xymatrix{
A
\ar[0,2]^-{h}
\ar[1,1]_{f}
&
\ar@{}[1,0]|(.4){\leq}
&
A'
\ar[1,-1]^{g}
\\
&
X
&
}  
}  
\quad
\mbox{implies}
\quad
\vcenter{
\xymatrixcolsep{1pc}
\xymatrix{
&
&
A'
\ar @/_1pc/ [1,-1]_{\Lan{h}{f}}
\ar@{} [1,-1]|{\leq}
\ar @/^1pc/ [1,-1]^{g}
\\
&
X
&
}  
}  
\end{equation}
\end{notation}

The following definition is due to Escard\'{o}~\cite{E1}
for objects and Carvalho and Sousa~\cite{CS} for morphisms:

\begin{definition}
\label{def:linj}
Let $h:A\to A'$ be a morphism of an order-enriched category.
\begin{enumerate}
\item
An object $X$ is called {\em left Kan-injective\/} w.r.t.
$h$ provided that for every morphism $f:A\to X$ there is a left
Kan extension $\Lan{h}{f}$ and it makes the following triangle
\begin{equation}
\label{eq:2.2}
\vcenter{
\xymatrixcolsep{1pc}
\xymatrix{
A
\ar[0,2]^-{h}
\ar[1,1]_{f}
&
&
A'
\ar[1,-1]^{\Lan{h}{f}}
\\
&
X
&
}  
}    
\end{equation}
commutative.
\item
A morphism $p:X\to X'$ is called {\em left Kan-injective\/}
w.r.t. $h$ if both $X$ and $X'$ are and for every $f:A\to X$
the morphism $p$ preserves the left Kan extension $\Lan{h}{f}$.
This means that the following diagram
\begin{equation}
\label{eq:2.3}
\vcenter{
\xymatrix{
A
\ar[0,1]^-{h}
\ar[1,0]_{f}
&
A'
\ar[1,-1]_{\Lan{h}{f}}
\ar[1,0]^{\Lan{h}{(pf)}}
\\
X
\ar[0,1]_-{p}
&
X'
}
}  
\end{equation}
commutes.
\end{enumerate}
\end{definition}

\begin{remark}
\mbox{}\hfill
\begin{enumerate}
\item 
Right Kan-injectivity is briefly mentioned in 
Section~\ref{sec:weak-kan-inj} below. (Escard\'{o} used
``right Kan-injective'' for left Kan-injectivity in~\cite{E1}.
We decided to follow the usual terminology, see, e.g., \cite{cwm}.)
\item
A weaker variant of left Kan-injectivity would just require
that for every $f$ the left Kan extension $\Lan{h}{f}$ exists
(i.e., we only have $f\leq \Lan{h}{f}\cdot h$, instead of equality).
We also turn to this concept in Section~\ref{sec:weak-kan-inj}, but we will
show that it can (under mild side conditions) be superseded by the
concept of Definition~\ref{def:linj}.
\end{enumerate}
\end{remark}

\begin{notation}
Let $\H$ be a class of morphisms of an order-enriched category $\X$.
We denote by 
$$
\LInj{\H}
$$  
the category of all objects and all morphisms that are left Kan-injective
w.r.t. all members of $\H$. The category $\LInj{\H}$ is order-enriched using
the enrichment of $\X$. 
\end{notation}

\begin{examples}
\label{exs:Pos}
We give examples of Kan-injectivity in $\Pos$. 
The order on homsets in $\Pos$ is defined pointwise.
\begin{enumerate}
\item 
\label{item:complete_semilattices}
{\em Complete semilattices\/}.
For $\H$ = all order-embeddings (that is, strong monomorphisms)
we have 
\begin{center}
$\LInj{\H}$ = complete join-semilattices and join-preserving maps.
\end{center}
Indeed, Bernhard Banaschewski and G\"{u}nter Bruns proved 
in~\cite{banaschewski+bruns} that every complete
(semi)lattice $X$ is left Kan-injective w.r.t. $\H$ since
for every order-embedding $h:A\to A'$ and every monotone $f:A\to X$
we have $\Lan{h}{f}$ given by
\begin{equation}
\label{eq:2.7new}
(\Lan{h}{f})(b)
=
\bigvee_{h(a)\leq b} f(a)
\end{equation}
And conversely, if $X$ is left Kan-injective, then every set
$M\subseteq X$ either has a maximum, which is $\bigvee M$,
or we have 
\begin{center}
$M\cap M^+ =\emptyset$ for $M^+$ = all upper bounds of $M$.   
\end{center}
In the latter case consider $A=M\cup M^+$ as a subposet of $X$
and let $A'$ extend $A$ by a single element $a'$ that is an upper
bound of $M$ and a lower bound of $M^+$. The embedding
$f:A\into X$ has a left Kan extension $\Lan{h}{f}$ that sends
$a'$ to $\bigvee M$.

By using the formula~\eqref{eq:2.7new} it is easy to see that a monotone map
$g:X\to Y$ between complete join-semilattices is left Kan-injective
iff $g$ preserves joins.
\item
{\em $\omega$CPOS's\/}.
Posets with joins of $\omega$-chains and $\bot$ and strict functions preserving
joins of $\omega$-chains are $\LInj{\H}$ for $\H$ consisting of the embeddings
$h:\omega\into\omega+1$ and $h':\emptyset\into 1$.
\item
\label{item:semilattices}
{\em Semilattices\/}.
For the embedding
$$
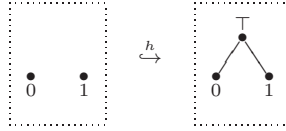
\let\objectstyle=\scriptstyle
\xy <1 pt,0 pt>:
    (000,000)  *++={};
    (030,040) *++={} **\frm{.};
    (070,000)  *++={};
    (100,040) *++={} **\frm{.}
\POS(005,015) *{\bullet};
    (025,015) *{\bullet};    
    (005,010) *{0} = "aA";
    (025,010) *{1} = "bA";
    (075,015) *{\bullet} = "a0B";
    (095,015) *{\bullet} = "b0B";    
    (075,010) *{0} = "aB";
    (095,010) *{1} = "bB";
    (085,030) *{\bullet} = "c0B";
    (085,035) *{\top} = "cB";
\POS(050,025) *{\stackrel{h}{\into}}
\POS "a0B" \ar@{-} "c0B";
\POS "b0B" \ar@{-} "c0B";
\endxy
$$
we obtain the category of join-semilattices and their
homomorphisms as $\LInj{\{ h\}}$. 
\item
{\em Conditional semilattices\/}.
For the embedding
$$
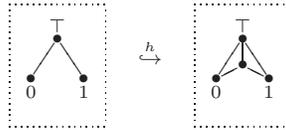
\let\objectstyle=\scriptstyle
\xy <1 pt,0 pt>:
    (000,000)  *++={};
    (030,040) *++={} **\frm{.};
    (070,000)  *++={};
    (100,040) *++={} **\frm{.}
\POS(005,015) *{\bullet} = "a0A";
    (025,015) *{\bullet} = "b0A";    
    (005,010) *{0} = "aA";
    (025,010) *{1} = "bA";
    (015,030) *{\bullet} = "c0A";
    (015,035) *{\top} ="cA";
    (075,015) *{\bullet} = "a0B";
    (095,015) *{\bullet} = "b0B";    
    (075,010) *{0} = "aB";
    (095,010) *{1} = "bB";
    (085,030) *{\bullet} = "c0B";
    (085,035) *{\top} = "cB";
    (085,020) *{\bullet} = "t0B";
    (085,015) *{} = "tB";
\POS(050,025) *{\stackrel{h}{\into}};
\POS "a0A" \ar@{-} "c0A";
\POS "b0A" \ar@{-} "c0A";
\POS "a0B" \ar@{-} "c0B";
\POS "b0B" \ar@{-} "c0B";
\POS "a0B" \ar@{-} "t0B";
\POS "b0B" \ar@{-} "t0B";
\POS "t0B" \ar@{-} "c0B";
\endxy
$$
we obtain the category of conditional join-semilattices
(where every pair with an upper bound has a join)
and maps that preserve nonempty finite joins as $\LInj{\{ h\}}$.
\item
{\em The category $\Pos_d$ of discrete posets\/}.
Form $\LInj{\{ h\}}$ for the morphism
$$
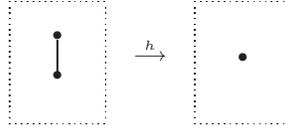
\let\objectstyle=\scriptstyle
\xy <1 pt,0 pt>:
    (000,000)  *++={};
    (030,040) *++={} **\frm{.};
    (070,000)  *++={};
    (100,040) *++={} **\frm{.}
\POS(015,015) *{\bullet} = "a0A";
    (015,030) *{\bullet} = "b0A";    
    (085,022) *{\bullet} = "a0B";
\POS(050,025) *{\stackrel{h}{\to}};
\POS "a0A" \ar@{-} "b0A";
\endxy
$$ 
\item
{\em The category $\Pos_1$ of posets of cardinality $\leq 1$\/}.
Form $\LInj{\{ h\}}$ for the mapping $h:1+1\to 1$.
\end{enumerate}
\end{examples}

Except for the trivial cases $\Pos_d$ and $\Pos_1$ all of the examples
in~\ref{exs:Pos} worked with $\H$ consisting of strong
monomorphisms. This is not coincidential: 

\begin{lemma}
Let $\H$ be a class of morphisms of $\Pos$ such that $\LInj{\H}$
is neither $\Pos_d$ nor $\Pos_1$. Then all members of $\H$ are
strong monomorphisms.  
\end{lemma}
\begin{proof}
Assume the contrary, i.e., suppose there exists $h:A\to A'$ in $\H$
such that for some $p$, $q$ in $A$ we have $h(p)\leq h(q)$ although
$p\nleq q$. Then we prove that every poset $X$ left Kan-injective
w.r.t. $h$ is discrete. It then follows easily that $\LInj{\H}$
is either $\Pos_d$ or $\Pos_1$.

Given elements $x\leq x'$ in $X$, we prove that $x=x'$.
Define $f:A\to X$ by
$$
f(a)=
\left\{
\begin{array}{rl}
x',&\mbox{ if $a\geq p$}\\
x, &\mbox{else}  
\end{array}
\right.
$$  
which is clearly monotone. Then $p\nleq q$ implies $f(q)=x$.
Consequently, $\Lan{h}{f}$ sends $h(p)$ to $x'$ and $h(q)$ to $x$.
Since $h(p)\leq h(q)$, we conclude $x'\leq x$, thus, $x=x'$.
\end{proof}

\begin{example}
\label{ex:scott}
The category $\Top_0$ of $T_0$ topological spaces and continuous maps
is order-enriched as follows. Recall the {\em specialisation order\/} 
$\sqsubseteq$ that Dana Scott~\cite{S72} used on every $T_0$-space:  
\begin{itemize}
\item[]
$x\sqsubseteq y$ iff every neighbourhood of $x$ contains $y$. 
\end{itemize}
We consider $\Top_0$ to be order-enriched by the opposite of the
pointwise specialisation order: for continuous functions $f,g:X\to Y$
we put 
\begin{itemize}
\item[]
$f\leq g$ iff $g(x)\sqsubseteq f(x)$ for all $x$ in $X$. 
\end{itemize}
\begin{enumerate}
\item 
\label{item:scott_1}
{\em Continuous lattices.\/}
For the collection $\H$ of all subspace embeddings in $\Top_0$
we have 
$$
\LInj{\H}
=
\mbox{ continuous lattices and meet-preserving continuous maps.}
$$
This was proved for objects by Escardo~\cite{EF} and for morphisms 
by Carvalho and Sousa~\cite{CS}, we present a proof for the convenience of the reader. 

Indeed, Scott proved that a $T_0$-space $X$ is injective iff
its specialisation order is a continuous lattice, i.e., a complete lattice
in which every element $y$ satisfies
\begin{equation}
\label{eq:scott_1}
y 
=
\bigsqcup_{U\in {\mathit{nbh}}(y)} 
\left( \bigsqcap U\right).
\end{equation}
Moreover, he gave, for every subspace embedding $h:A\to A'$
and every continuous map $f:A\to X$, a concrete formula 
for a continuous extension $f':A'\to X$:
\begin{equation}
\label{eq:scott_2}
f'(a')
=
\bigsqcup_{U\in {\mathit{nbh}}(a')}  
\left( \bigsqcap f(h^{-1}(U))\right)
\mbox{ for all $a'\in A'$}.
\end{equation}
This is actually the desired left Kan extension
$f'=\Lan{h}{f}$, as proved by Escard\'{o}~\cite{E1}.
His proof uses the filter monad $\mathcal{F}$
on $\Top_0$ whose Eilenberg-Moore algebras are,
as proved by Alan Day~\cite{day} and Oswald Wyler~\cite{wyler},
precisely the continuous lattices: for every 
continuous lattice $X$ the algebra $\alpha:{\mathcal{F}}X\to X$ 
is defined by
\begin{equation}
\label{eq:scott_3}
\alpha(F)
=
\bigsqcup_{U\in F} 
\left( \bigsqcap U\right)
\mbox{ for all filters $F$}.  
\end{equation}
Every continuous map $p:X\to Y$ between continuous lattices preserving
meets is Kan-injective. This follows from the formula~\eqref{eq:scott_2}
for $\Lan{h}{f}$: given $f:A\to X$ we have
$$
\begin{array}{rcll}
p\cdot (\Lan{h}{f})(a')
&=&
p
\left(
\displaystyle\bigsqcup_{U\in {\mathit{nbh}}(a')}  
\left( \bigsqcap f(h^{-1}(U))\right)
\right)  
&
\mbox{by~\eqref{eq:scott_2}}
\\
&=&
\displaystyle\bigsqcup_{U\in {\mathit{nbh}}(a')}  
p\left( \bigsqcap f(h^{-1}(U))\right)
&
\mbox{since $p$ is continuous}
\\
&=&
\displaystyle\bigsqcup_{U\in {\mathit{nbh}}(a')}  
\left( \bigsqcap pf(h^{-1}(U))\right)
&
\mbox{since $p$ preserves meets}
\\
&=&
\Lan{h}{(pf)}(a')
&
\mbox{by~\eqref{eq:scott_2}}
\end{array}
$$
Conversely, if a continuous map $p:X\to Y$ is Kan-injective,
then it preserves meets. Indeed, following Day, $p$
is a homomorphism of the corresponding monad algebras. Given
$M\subseteq X$, let $F_M$ be the filter of all subsets 
containing $M$, then~\eqref{eq:scott_3} yields 
$\alpha(F_M)=\bigsqcap M$ --- hence, the fact that $p$ is 
a homomorphism implies that $p$ preserves meets.
\item
{\em Continuous Scott Domains.\/}
For the collection $\H$ of all dense subspace embeddings
we have
$$
\LInj{\H}
=
\mbox{ continuous Scott domains and continuous functions preserving
nonempty meets.}
$$
Recall that a continuous Scott domain is a poset with bounded joins
(or, equivalently, nonempty meets) satisfying~\eqref{eq:scott_1}.
Escard\'{o} proved that the $T_0$ spaces Kan-injective w.r.t.
dense embeddings are precisely those whose order is a continuous
Scott domain. His proof uses the monad ${\mathcal{F}}^+$ of proper
filters on $\Top_0$. The conclusion that Kan-injective morphisms
are precisely those preserving nonempty meets is analogous 
to~\eqref{item:scott_1}.
\end{enumerate}
\end{example}

\begin{remark}
The order enrichment of $\Top_0$ above is frequently used in literature.
However, some authors prefer the dual enrichment (by the pointwise
specialisation order). We mention in Example~\ref{ex:E-right} 
below that this yields the same examples as above but for
the right Kan-injectivity. 
\end{remark}

\begin{example}
\label{ex:orthogonal}
Given an ordinary category, we can consider it order-enriched
by the trivial order. An object $X$ is then Kan-injective
w.r.t. $\H$ iff it is {\em orthogonal\/}, i.e., given $h:A\to A'$  
it fulfills: for every $f:A\to X$ there is a unique $f':A'\to X$
such that the triangle
$$
\xymatrixcolsep{1pc}
\xymatrix{
A
\ar[0,2]^-{h}
\ar[1,1]_{f}
&
&
A'
\ar[1,-1]^{f'}
\\
&
X
&
}
$$
commutes.

And every morphism between orthogonal objects is Kan-injective.
Thus, the Kan-injectivity subcategory is precisely
$$
\H^\perp
=
\LInj{\H}
$$
the full subcategory of all orthogonal objects.
\end{example}

\begin{remark}
\label{rem:idempotent-monads}
\mbox{}\hfill
\begin{enumerate}
\item 
A special case is given by a monad $\TT=(T,\eta,\mu)$
on the (ordinary) category which is {\em idempotent\/},
i.e., fulfills 
$$
T\eta =\eta T
$$
Consequently, every object $X$ carries at most one structure
on an Eilenberg-Moore algebra $x:TX\to X$, since $x=\eta_X^{-1}$.
Thus, the category $\X^\TT$ can be considered as a full subcategory
of $\X$. For the class $\H=\{\eta_X\mid \mbox{$X$ in $\X$}\}$
of all units of $\TT$ we then have
$$
\X^\TT
=
\H^\perp
$$
\item
Conversely, whenever the full subcategory $\H^\perp$ is
reflective, i.e., its embedding into $\X$ has a left adjoint,
then the corresponding monad $\TT$ on $\X$ is idempotent
and $\X^\TT\cong\H^\perp$. 
\item
The concepts of (i)~full reflective subcategory of $\X$,
(ii)~idempotent monad on $\X$ and (iii)~orthogonal subcategory
$\H^\perp$ coincide --- modulo the {\em orthogonal subcategory problem\/}.
This is the problem whether given a class $\H$ of morphisms
the subcategory $\H^\perp$ is reflective. Some positive solutions
can be found in~\cite{freyd+kelly} and~\cite{AHS2}, for a negative solution
in $\X = \Top$ see~\cite{ar88}.  
\end{enumerate}
\end{remark}

The situation with order-enriched categories is completely
analogous, as we prove below. The following can be found 
in~\cite{E1} and~\cite{CS}.

\begin{example}
\label{ex:KZ=lan-injective}
Let $\TT=(T,\eta,\mu)$ be a {\em Kock-Z\"{o}berlein monad\/}
on an order-enriched category $\X$, i.e., one satisfying
$$
T\eta\leq \eta T.
$$  
Kock-Z\"{o}berlein monads over order-enriched categories are a particular
case of the monads on 2-categories, independently introduced 
by Anders Kock~\cite{K} and Volker Z\"{o}berlein~\cite{zoeberlein}.

Every object $X$ carries at most one structure of an Eilenberg-Moore
algebra $\alpha:TX\to X$, since $\alpha$ is left adjoint to $\eta_X$.
Thus, $\X^\TT$ can be considered as a (not necessarily full) subcategory 
of $\X$. Then the category of $\TT$-algebras consists
precisely of all objects and morphisms Kan-injective
to all units:
$$
\X^\TT
=
\LInj{\H}
\mbox{ for $\H=\{\eta_X\mid \mbox{$X$ in $\X$}\}$}
$$
see Proposition~\ref{prop:4.5} below. Conversely, whenever the subcategory $\LInj{\H}$
is reflective, i.e., its (possibly non-full) embedding into
$\X$ has a left adjoint, then it is monadic and the corresponding
monad $\TT$ satisfies the Kock-Z\"{o}berlein
property, see Corollary~\ref{cor:reflective=>KZ-monadic} below.
\end{example}

\section{Inserters and coinserters}
\label{sec:inserters+coinserters}

Since inserters and coinserters play a central role in our paper,
we recall the facts about them we need (in our special case of
order-enriched categories) in this section.
Throughout this section we work in an order-enriched category.

\begin{definition}
\label{def:inserter}
\mbox{}\hfill
\begin{enumerate}
\item
We call a morphism $i:I\to X$ an {\em order-monomorphism\/}
provided that for all $f,g:I'\to I$ we have: $i\cdot f\leq i\cdot g$
implies $f\leq g$.
\item
An {\em inserter\/} of a parallel pair $u,v:X\to Y$
in an order-enriched category is a morphism $i:I\to X$
universal w.r.t. $u\cdot i\leq v\cdot i$.
$$
\xymatrix{
I
\ar[0,1]^-{i}
&
X
\ar@<.5ex>[0,1]^-{u}
\ar@<-.5ex>[0,1]_-{v}
&
Y
\\
J
\ar[-1,0]^{\ol{j}}
\ar[-1,1]_{j}
&
}
$$  
Universality means the following two conditions:
\begin{enumerate}
\item 
Given $j$ with $u\cdot j\leq v\cdot j$, there exists a unique
$\ol{j}$ with $j=i\cdot\ol{j}$.
\item
$i$ is an order-monomorphism.
\end{enumerate}  
\end{enumerate}
\end{definition}

\begin{example}
In $\Top_0$ the inserter of $u,v:X\to Y$ is the embedding
$I\into X$ of the subspace of $X$ on all elements $x\in X$
with $u(x)\leq v(x)$. In general, every subspace embedding is
an order-monomorphism.

In $\Pos$, analogously, the inserter of $u,v:X\to Y$ 
is the embedding $I\into X$ of the subposet of $X$
on all elements $x\in X$ with $u(x)\leq v(x)$.
In general, every subposet embedding is an 
order-monomorphism --- and vice versa (up to isomorphism).  
\end{example}

\begin{lemma}
For a morphism $i$ in $\Pos$ the following conditions
are equivalent:
\begin{enumerate}
\item 
$i$ is an order-monomorphism.
\item
$i$ is a strong monomorphism.
\item
$i$ is a subposet embedding (up to isomorphism).
\item
$i$ is an inserter of some pair.
\end{enumerate}
\end{lemma}
\begin{proof}
It is easy to see that~(2) and~(3) are both equivalent to the
validity of the implication ``$i(x)\leq i(y)$ implies $x\leq y$''.
Therefore~(1) implies~(3). To prove~(3) implies~(4),
given a subposet embedding $i:X\into Y$, let $Z$ be the 
poset obtained from $Y$ by splitting every element outside 
of $i[X]$ to two incomparable elements. The two obvious 
embeddings of $Y$ into $Z$ have $i$ as their inserter.
Finally, (4) implies~(1) by the definition.  
\end{proof}

\begin{definition}
\label{def:coinserter}
\mbox{}\hfill
\begin{enumerate}
\item 
An {\em order-epimorphism\/} is a morphism $e:X\to Y$
such that for all $f,g:Y\to Z$ we have: $f\cdot e\leq g\cdot e$
implies $f\leq g$.
\item
A {\em coinserter\/} of a parallel pair $u,v:X\to Y$
is a morphism $c:Y\to C$ couniversal w.r.t. $c\cdot u\leq c\cdot v$.
That is, the following two conditions hold:
\begin{enumerate}
\item 
Given $d:Y\to Z$ with $d\cdot u\leq d\cdot v$ there exists a unique
$\ol{d}:C\to Z$ with $d=\ol{d}\cdot c$.
\item
$c$ is an order-epimorphism. 
\end{enumerate}
\end{enumerate}
\end{definition}

\begin{examples}
\label{exs:coinserters}
\mbox{}\hfill
\begin{enumerate}
\item 
In $\Pos$ every surjection (= epimorphism)
is an order-epimorphism, see Lemma~\ref{lem:epis-in-Pos} below. 
\item
In $\Top_0$ also every epimorphism is an order-epimorphism.
We can describe coinserters by using those in $\Pos$ and 
applying the forgetful functor 
$$
U:\Top_0\to\Pos
$$
of Example~\ref{ex:scott}.

This functor has the following universal property: given a monotone
function $c:UY\to (Z,\leq)$ where $Y$ is a $T_0$ space, there exists
a {\em semifinal solution\/} in the sense of~25.7~\cite{AHS}, which means
a pair consisting of $\ol{c}:Y\to\ol{Z}$ in $\Top_0$ and
$c_0:(Z,\leq)\to U\ol{Z}$ in $\Pos$ universal w.r.t.
$$
\xymatrixcolsep{1pc}
\xymatrix{
UY
\ar[0,2]^-{U\ol{c}}
\ar[1,1]_{c}
&
&
U\ol{Z}
\\
&
(Z,\leq)
\ar[-1,1]_{c_0}
&
}
$$
Thus given another pair $\wt{c}:Y\to\wt{Z}$ and $\wt{c_0}:(Z,\leq)\to U\wt{Z}$
with $U\wt{c}=\wt{c_0}\cdot c$ there exists a unique $p:\ol{Z}\to\wt{Z}$
in $\Top_0$ making the diagrams
$$
\vcenter{
\xymatrixcolsep{1pc}
\xymatrix{
Y
\ar[0,2]^-{\wt{c}}
\ar[1,1]_{\ol{c}}
&
&
\wt{Z}
\\
&
\ol{Z}
\ar[-1,1]_{p}
&
}
}
\quad
\mbox{and}
\quad
\vcenter{
\xymatrixcolsep{1pc}
\xymatrix{
(Z,\leq)
\ar[0,2]^-{\wt{c_0}}
\ar[1,1]_{c_0}
&
&
U\wt{Z}
\\
&
U\ol{Z}
\ar[-1,1]_{Up}
&
}
}
$$
commutative.

Indeed, to construct $\ol{c}$, let $\tau$ be the topology
on $Z$ of all lowersets whose inverse image under $c$ is open
in $Y$. Let $r:(Z,\tau)\to \ol{Z}$ be a $T_0$-reflection,
then put $\ol{c}=r\cdot c$. Consequently, we see that each such
$\ol{c}$ is an order-epimorphism in $\Pos$.

The coinserter of $u,v:X\to Y$ in $\Top_0$ is obtained
by first forming a coinserter $c:UY\to (Z,\leq)$ of 
$Uu$, $Uv$ in $\Pos$ and then taking the semifinal solution
$\ol{c}:Y\to\ol{Z}$.
\end{enumerate}
\end{examples}

\begin{lemma}
\label{lem:epis-in-Pos}
For a morphism $e$ in $\Pos$ the following conditions 
are equivalent:  
\begin{enumerate}
\item 
$e$ is an order-epimorphism.
\item
$e$ is an epimorphism.
\item
$e$ is surjective.
\item
$e$ is a coinserter of some pair.
\end{enumerate}
\end{lemma}
\begin{proof}
The equivalence of (2) and (3) is well-known, see,
e.g., Example~7.40(2)~\cite{AHS}.

It is clear that (1) implies~(2) and~(4) implies~(1).  
To prove that (3) implies~(4), choose a surjective map
$e:A\to B$ and define the poset $A_0$ as follows:
its elements are pairs $(x,x')$ such that $e(x)\leq e(x')$,
the pairs are ordered pointwise. Denote by 
$d_0,d_1:A_0\to A$ the obvious monotone 
projections. Then it follows easily that $e$ is a coinserter of the
pair $(d_0,d_1)$, using the fact that $e$ is surjective.
\end{proof}

\begin{definition}
An order-enriched category is said to have
{\em conical products\/} if it has products $\prod_{i\in I}X_i$
and the projections $\pi_i$ are collectively order-monic.
That is, given a parallel pair $f,g:Y\to\prod_{i\in I} X_i$
we have that
\begin{equation}
\label{eq:3.1new}
\pi_i\cdot f\leq\pi_i\cdot g
\mbox{ for all $i\in I$ implies }
f\leq g.  
\end{equation}
\end{definition}

\begin{example}
In $\Top_0$ and $\Pos$ products are clearly conical.  
\end{example}

\begin{remark}
Throughout Section~\ref{sec:inserter-ideals} 
we work with order-enriched
categories having inserters and conical products.
This can be expressed more compactly by saying that
weighted limits exist. We recall this fact (that can
be essentially found in Max Kelly's book~\cite{kelly:book})
for convenience of the reader. However, we are not going
to apply any weighted limits except inserters and conical
limits in our paper.  
\end{remark}

Given order-enriched categories $\X$ and $\D$,
where $\D$ is small, we denote by
$$
\X^\D
$$
the order-enriched category of all locally monotone
functors from $\D$ to $\X$ and all natural transformations
between them (the order on natural transformations is
objectwise: given $\alpha,\beta:F\to G$ then $\alpha\leq\beta$
means $\alpha_d\leq\beta_d$ for every $d$ in $\D$).

\begin{definition}
Let $\X$ and $\D$ be order-enriched categories, $\D$ small.
Given a locally monotone functor $D:\D\to\X$, its
{\em limit weighted by\/} $W:\D\to\Pos$, also locally
monotone, is an object
$$
\Lim{W}{D}
$$
together with an isomorphism
$$
\X(X,\Lim{W}{D})
\cong
\Pos^\D(W,\X(X,D-))
$$  
natural in $X$ in $\X$.
\end{definition}

\begin{examples}
\label{exs:3.9new}
\mbox{}\hfill
\begin{enumerate}
\item 
{\em Conical limits\/} (which means limits whose limit cones 
fulfill~\eqref{eq:3.1new}) 
are precisely the weighted limits
with weight constantly $1$ (the terminal poset).
\item
Inserters are weighted limits with the scheme
$$
\xy <1 pt,0 pt>:
    (000,000)  *++={};
    (040,020) *++={} **\frm{.};
\POS(010,010) *{\bullet} = "l";
    (030,010) *{\bullet} = "r";
    "l"+(-07,000) *{d}; 
    "r"+(007,000) *{d'};
    (-15,010) *{\D:}    
\POS "l"+(002,005) \ar@{->}^-{v} "r"+(-02,005)
\POS "l"+(002,-05) \ar@{->}_-{u} "r"+(-02,-05)
\endxy
$$  
and the weight $W$ given by
$$
\xy <1 pt,0 pt>:
    (000,005)  *++={};
    (010,015) *++={} **\frm{.};
    (040,000)  *++={};
    (050,020) *++={} **\frm{.};
\POS(005,010) *{\bullet} = "l";
    (045,005) *{\bullet} = "rd";
    (045,015) *{\bullet} = "ru";
\POS "l"+(003,001) \ar@{->}^-{Wv} "ru"+(-03,001)
\POS "l"+(003,-01) \ar@{->}_-{Wu} "rd"+(-03,-01)
\POS "rd" \ar@{-} "ru" 
\endxy
$$  
\end{enumerate}
\end{examples}

\begin{remark}
\label{rem:products+equalisers=limits}
A category with conical products and inserters has 
conical equalisers, hence all conical limits.
Indeed, an equaliser of a pair $f,g:X\to Y$
is obtained as an inserter of the pair
$$
\xymatrix{
X
\ar@<.5ex>[0,1]^-{\langle f,g\rangle}
\ar@<-.5ex>[0,1]_-{\langle g,f\rangle}
&
X\times Y
}
$$  
Just observe that a morphism $i:I\to X$ fulfills 
$\langle f,g\rangle\cdot i\leq \langle g,f\rangle\cdot i$
iff it fulfills $f\cdot i=g\cdot i$. Moreover, we see that
equalisers are order-monomorphisms (since inserters are).
\end{remark}

\begin{lemma}
\label{lem:products+inserters=limits}
An order-enriched category has weighted limits 
iff it has conical products and inserters.  
\end{lemma}
\begin{proof}
The necessity follows from Examples~\ref{exs:3.9new}.  
For the sufficiency, we use Theorem~3.73 of~\cite{kelly:book}.
In fact, it suffices to prove that a particular
type of weighted limits, called {\em cotensors\/}, exists in $\X$. 
Given a poset $P$ and an object $X$, then the $P$-th cotensor of 
$X$ is an object $P\cotensor X$, together with an isomorphism
$$
\X(X',P\cotensor X)
\cong
\Pos(P,\X(X',X))
$$  
natural in $X'$.

Observe that, for a discrete poset $P$, the cotensor 
$P\cotensor X$ is just the $P$-fold conical product of $X$.
Hence the category $\X$ has cotensors with discrete
posets, since it has products.

A general poset $P$ can be described as a coinserter in $\Pos$ 
of a parallel pair
$$
\xymatrix{
P_1
\ar@<.5ex>[0,1]^-{d_1}
\ar@<-.5ex>[0,1]_-{d_0}
&
P_0
}
$$
where $P_0$ is the discrete poset on elements of $P$, $P_1$
is the discrete poset on all pairs $(x,x')$ such that $x\leq x'$
holds, and $d_0$ and $d_1$ are the obvious projections.
Then one can define $P\cotensor X$ as an inserter of 
$$
\xymatrix{
P_0\cotensor X
\ar@<.5ex>[0,1]^-{d_1\cotensor X}
\ar@<-.5ex>[0,1]_-{d_0\cotensor X}
&
P_1\cotensor X
}
$$
in $\X$.
\end{proof}

Whereas inserters and conical products are required in 
Section~\ref{sec:inserter-ideals}, we work with the 
dual concepts in Section~\ref{sec:chain}. 

\begin{definition}
An order-enriched category is said to have {\em conical coproducts\/}
if it has coproducts $\coprod_{i\in I} X_i$ and the injections
$\gamma_i$ are collectively order-epic. That is, given a parlallel pair 
$f,g:\coprod_{i\in I} X_i\to Y$, we have that $f\cdot\gamma_i\leq g\cdot\gamma_i$
for all $i\in I$ implies $f\leq g$.  
\end{definition}

\begin{example}
\label{ex:coproducts}
The categories $\Pos$ and $\Top_0$ clearly have conical coproducts.
Therefore, they have conical colimits. This is dual to 
Remark~\ref{rem:products+equalisers=limits}.  
\end{example}

Again, the dual notions can
be subsumed by the concept of a weighted colimit.

\begin{definition}
Let $\X$ and $\D$ be order-enriched categories, $\D$ small.
Given a locally monotone functor $D:\D\to\X$, its
{\em colimit weighted by\/} $W:\D^\op\to\Pos$, also locally
monotone, is an object
$$
\Colim{W}{D}
$$
together with an isomorphism
$$
\X(\Colim{W}{D},X)
\cong
\Pos^{\D^\op}(W,\X(D-,X))
$$  
natural in $X$ in $\X$.  
\end{definition}

\begin{lemma}
\label{lem:coproducts+coinserters=colimits}
An order-enriched category has weighted colimits iff it has
conical coproducts and coinserters.  
\end{lemma}
\begin{proof}
This is dual to Lemma~\ref{lem:products+inserters=limits}.  
\end{proof}

\section{KZ-monadic subcategories and inserter-ideals}
\label{sec:inserter-ideals}

In this section we prove that whenever the Kan-injectivity 
subcategory $\LInj{\H}$ is reflective, then the monad $\TT$ this
generates is a Kock-Z\"{o}berlein monad and the Eilenberg-Moore
category $\X^\TT$ is precisely $\LInj{\H}$. In the subsequent
sections we prove that for small collections $\H$ in
``reasonable'' categories $\LInj{\H}$ is always reflective.
A basic concept we need is that of an inserter-ideal subcategory.

\begin{definition}
A subcategory of an order-enriched category $\X$ 
is {\em inserter-ideal\/}
provided that it contains with every morphism $u$
also inserters of the pairs $(u,v)$, where $v$
is any morphism in $\X$ parallel to $u$.  
\end{definition}

\begin{lemma}
\label{lem:inserter-ideal}
Every Kan-injectivity subcategory $\LInj{\H}$ is inserter-ideal.  
\end{lemma}
\begin{proof}
Suppose that we have an inserter 
$i$ of $(u,v)$ in $\X$.
It is our task to prove that if $u$
is left Kan-injective w.r.t. $h:A\to A'$ in $\H$,
then so is $i$. We first verify that $I$ is
left Kan-injective. Consider an arbitrary $f:A\to I$.
In the following diagram
$$
\xymatrix{
A
\ar[0,2]^-{h}
\ar[1,0]_{f}
&
&
A'
\ar[1,0]^{\Lan{h}{(if)}}
\ar@{.>}[1,-2]_{f^*}
&
\\
I
\ar[0,2]_-{i}
&
&
X
\ar@<.5ex>[0,1]^-{v}
\ar@<-.5ex>[0,1]_-{u}
&
Y
Y
}
$$  
the morphism $\Lan{h}{(if)}:A'\to X$ exists since
$X$ is left Kan-injective. Also, $u$ is left Kan-injective
and therefore we have
$$
u\cdot \Lan{h}{(if)} 
=
\Lan{h}{(uif)}
\leq
\Lan{h}{(vif)}
\leq
v\cdot\Lan{h}{(if)} 
$$
proving that $\Lan{h}{(if)}$ factorises through $i$ 
as indicated above. 

That the morphism $f^*:A'\to I$ is $\Lan{h}{f}$
follows immediately from the two aspects of the universal 
property of an inserter. This proves that the object $I$ 
is left Kan-injective w.r.t. $h$. 

Moreover, we also have the 
equality $\Lan{h}{(if)}=i\cdot f^*=i\cdot\Lan{h}{f}$, proving
that the morpism $i:I\to X$ is left Kan-injective w.r.t.
$h$, as desired.
\end{proof}

\begin{corollary}
$\LInj{\H}$ is closed under weighted limits.  
\end{corollary}
\begin{proof}
Indeed, it is closed under inserters by Lemma~\ref{lem:inserter-ideal}
and under conical limits by~\cite{CS}, Proposition~2.10.  
The rest is analogous to the proof of
Lemma~\ref{lem:products+inserters=limits} above.
\end{proof}

\begin{definition}
\label{def:KZ-monadic}
A subcategory of an order-enriched category $\X$ is called
{\em KZ-monadic\/} if it is the Eilenberg-Moore category
$\X^\TT$ of a Kock-Z\"{o}berlein monad $\TT$ on $\X$. 
\end{definition}

\begin{example}
\label{ex:KZ-monads}
\mbox{}\hfill
\begin{enumerate}
\item
Continuous lattices, see Example~\ref{ex:scott}\eqref{item:scott_1},
are KZ-monadic for the filter monad on $\Top_0$, as proved by
Escard\'{o}~\cite{E1}.  
\item
Complete semilattices, see Example~\ref{exs:Pos}\eqref{item:complete_semilattices},
are KZ-monadic w.r.t. the lowerset monad  $\TT=(T,\eta,\mu)$ on $\Pos$.
More in detail: $TX$ is the poset of all lowersets on a poset $X$,
$\eta_X:X\to TX$ assigns the principal lowerset ${\downarrow}x$ 
to every $x\in X$, $\mu_X:TTX\to TX$ is the union.
\end{enumerate}
\end{example}

\begin{remark}
\label{rem:left-adj-retr}
Recall the concept of a projection-embedding pair
of Mike Smyth and Gordon Plotkin~\cite{smyth+plotkin}. We use the dual concept and 
call a morphism $r:C\to X$ a {\em coprojection\/} if there exists $s:X\to C$ with 
$$
r\cdot s=\id_C 
\mbox{ and } 
\id_X\leq s\cdot r.
$$ 
In the terminology of~\cite{CS} the morphism $r$ would be called reflective left adjoint.
\end{remark}

\begin{definition}
\label{def:left-adj-retr}
A subcategory  $\C$ of an order-enriched category $\X$ is said
to be {\em closed under coprojections\/} if
(a)~for every coprojection $r:C\to X$ whenever $C$ is in $\C$,
then so is $X$, and (b)~for any commutative square in $\X$
$$
\xymatrix{
C_1
\ar[r]^f
\ar[d]_{r_1}
&
C_2
\ar[d]^{r_2}
\\
X_1
\ar[r]_g
&
X_2
}
$$
whenever $f$ is in $\C$ and $r_1$, $r_2$ are coprojections, 
then also $g$ is in $\C$.
\end{definition}

\begin{proposition}[Proposition~2.13 of~\cite{CS}]
\label{prop:coprojections}
Every Kan-injectivity subcategory $\LInj{\H}$ is closed
under coprojections.  
\end{proposition}

\begin{proposition}[See~\cite{bunge+funk} and~\cite{CS}]
\label{prop:4.5}
Every KZ-monadic category is the
Kan-injectivity subcategory w.r.t. all units, i.e.,
$$
\X^\TT = \LInj{\H}
\; 
\mbox{ for }
\; 
\H = \{\eta_X:X\to TX \mid \mbox{$X$ in $\X$}\}.
$$
\end{proposition}	

This follows from Proposition~1.5 and Corollary~1.6 in~\cite{bunge+funk},
as well as from Theorem~3.9 and Remark~3.10 in~\cite{CS}.

\begin{remark}
For the larger collection
$\H'$ of all morphisms $i$ with $Ti$ having a right adjoint
$Ti\dashv j$ such that $j\cdot Ti=\id$ it also holds that 
$\X^\TT=\LInj{\H'}$, see~\cite{EF} and~\cite{CS}.   
\end{remark}

\begin{theorem}
\label{th:reflective=>KZ-monadic}
A subcategory of an order-enriched category is KZ-monadic iff it is
\begin{enumerate}
\item reflective,
\item inserter-ideal, and
\item closed under coprojections.
\end{enumerate}
\end{theorem}
\begin{proof}
We first recall from~\cite{CS}, Theorems~3.13 and~3.4 that a subcategory
$\C$ is KZ-monadic iff it is
\begin{enumerate}
\renewcommand{\theenumi}{\alph{enumi}}
\item 
reflective, with reflections $\eta_X:X\to FX$ ($X$ in $\X$)
\item
closed under coprojections,
\item
a subcategory of $\LInj{\H}$ for $\H=\{\eta_X\mid\mbox{$X$ in $\X$}\}$,
\end{enumerate}
and such that
\begin{enumerate}
\renewcommand{\theenumi}{\alph{enumi}}
\setcounter{enumi}{3}
\item 
every morphism $f:FX\to A$ in $\C$ fulfils $\Lan{\eta_X}{(f\eta_X)}=f$.
\end{enumerate}
Indeed, Theorem~3.4 states that (a), (c) and (d) are equivalent to $\C$
being KZ-reflective, thus Theorem~3.13 applies.

Every KZ-monadic category is inserter ideal by Lemma~\ref{lem:inserter-ideal} 
and Proposition~\ref{prop:4.5}, thus it has all the properties of our
Theorem: see Conditions~(a) and~(b) above.

For the converse implication, we only need to verify
Conditions~(c) and~(d) above. For~(c) see Proposition~\ref{prop:4.5}. 
Condition~(d) easily follows from the implication
$$
\mbox{$f\eta_X\leq g\eta_X$ implies $f\leq g$}
$$
for all pairs $f,g:FX\to A$ with $f$ in $\C$.

In order to prove the implication,
form the inserter $i$ of the pair $(f,g)$:
$$
\xymatrix{
I
\ar@<.5ex>[0,1]^-{i}
&
FX
\ar@<.5ex>[0,1]^-{g}
\ar@<-.5ex>[0,1]_-{f}
\ar@<.5ex>[0,-1]^-{v}
&
C
\\
X
\ar[-1,1]_{\eta_X}
\ar[-1,0]^{u}
&
}
$$
Thus, we have a morphism $u$ in $\X$ with $\eta_X=i\cdot u$.
Since $f$ lies in the inserter-ideal subcategory $\C$,
so does $I$. Therefore $u$ factorises through the reflection
$\eta_X$:
$$
u=v\cdot\eta_X
$$  
and both $v$ and $i$ are morphisms of $\C$.
Thus so is $i\cdot v$ and from $(i\cdot v)\cdot\eta_X=\eta_X$
we therefore conclude $i\cdot v=\id$. Now $i$ is monic as well as
split epic, therefore it is invertible. This gives the desired
inequality $f\leq g$.
\end{proof}

From Lemma \ref{lem:inserter-ideal}, Theorem \ref{th:reflective=>KZ-monadic}, 
and Proposition~\ref{prop:coprojections}, we obtain the following:

\begin{corollary}
\label{cor:reflective=>KZ-monadic}
Whenever $\LInj{\H}$ is a reflective subcategory, then it is
KZ-monadic.  
\end{corollary}

\section{Kan-injective reflection chain}
\label{sec:chain}

Here we show how a reflection of an object $X$ in the 
Kan-injectivity subcategory $\LInj{\H}$ is constructed:
we define a transfinite chain $X_i$ ($i\in\Ord$) with $X_0=X$
such that with increasing $i$ the objects $X_i$ are ``nearer''
to being Kan-injective. This chain is said to {\em converge\/}
if for some ordinal $k$ the connecting map
$\xymatrix@1{X_k\ar@{-->}[r]&X_{k+2}}$ is invertible.
When this happens, $X_k$ is Kan-injective, and a reflection
of $X$ is given by the connecting map 
$\xymatrix@1{X_0\ar@{-->}[r]&X_k}$. In Section~\ref{sec:locally-ranked-cats}
sufficient conditions for the convergence of the reflection chain
are discussed.

\begin{assumption}
Throughout this section $\X$ denotes an order-enriched
category with weighted colimits.  
\end{assumption}

\begin{construction}[\bf Kan-injective reflection chain]
\label{cons:reflection}
Let $\X$ be an order-enriched category with weighted colimits, and
$\H$ a set of morphisms in $\X$. Given an object $X$, we 
construct a chain of objects $X_i$ ($i\in \Ord$).
We denote the connecting maps by $x_{ij}:X_i\to X_j$ or just by
$\xymatrix@1{X_i\ar@{-->}[r]&X_j}$,  for all $i\leq j$.  

The first step is the given object $X_0=X$.
Limit steps $X_i$, $i$ a limit ordinal, are defined by (conical) 
colimits of $i$-chains:
$$
X_i=\colim_{j<i}X_j.
$$
Isolated steps: given $X_i$ we define both $X_{i+1}$ and $X_{i+2}$, 
thus, we can restrict ourselves to even ordinals $i$ (having distance $2n$, 
$n<\omega$, from $0$ or a limit ordinal).
\begin{enumerate}
\item 
To define $X_{i+1}$ and the connecting map 
$\xymatrix@1{X_i\ar@{-->}[r]&X_{i+1}}$, consider all spans
\begin{equation}
\label{eq:4.1}
\vcenter{
\xymatrix{
A
\ar[d]_f
\ar[r]^h
&
A'
\\
X_i
&
}
}
\end{equation}
where $h$ is in $\H$ and $f$ is arbitrary. We form the  colimit 
of this diagram and call the colimit morphisms 
$\xymatrix@1{X_i\ar@{-->}[r]&X_{i+1}}$ and $f\dd h$ 
(because they ``approximate'' $\Lan{h}{f}$), respectively:
\begin{equation}
\label{eq:4.2}
\vcenter{
\xymatrix{
A
\ar[d]_f
\ar[r]^h
&
A'
\ar[d]^{f\dd h}
\\
X_i
\ar@{-->}[r]
&X_{i+1}
}
}
\end{equation}
More detailed: given $h$ in $\H$ and $f:A\to X_i$ we form a pushout
\begin{equation}
\label{eq:5.2a}
\vcenter{
\xymatrix{
A
\ar[0,1]^-{h}
\ar[1,0]_{f}
&
A'
\ar[1,0]^{\ol{f}}
\\
X_i
\ar[0,1]_-{\ol{h}}
&
C
}
}
\end{equation}
Then $\xymatrix@1{X_i\ar@{-->}[r]&X_{i+1}}$ is the wide pushout 
of all $\ol{h}$ (with the colimit cocone $c_{f,h}:C\to X_{i+1}$)
and we put $f\dd h=c_{f,h}\cdot \ol{f}$.
\item
\label{item:construction-2}
To define $X_{i+2}$ and the connecting map 
$\xymatrix@1{X_{i+1}\ar@{-->}[r]&X_{i+2}}$, consider all inequalities
\begin{equation}
\label{eq:4.3}
\vcenter{
\xymatrix{
A
\ar[d]_{f}
\ar[r]^h
&
A'\ar[d]^g
\\
X_j
\ar@{}[-1,1]|{\leq}
\ar@{-->}[r]
&
X_{i+1}
}
}
\end{equation}
where $h\in\H$, $j\leq i$ is an even ordinal, and $f$, $g$ are arbitrary. 
We let $\xymatrix@1{X_{i+1}\ar@{-->}[r]&X_{i+2}}$ be the universal map 
such that~\eqref{eq:4.3} implies the inequality
\begin{equation}
\label{eq:4.4}
\vcenter{
\xymatrix{
&
A'
\ar[dl]_{f\dd h}
\ar[ddr]^g
\ar@{}[3,0]|{\textstyle\leq}
&
\\
X_{j+1}
\ar@{-->}[d]
&
&
\\
X_{i+1}\ar@{-->}[dr]
&
&
X_{i+1}\ar@{-->}[dl]
\\
&
X_{i+2}
&
}
}
\end{equation}
In other words, $\xymatrix@1{X_{i+1}\ar@{-->}[r]&X_{i+2}}$ is the wide 
pushout of all the coinserters 
$$
\coins(x_{j+1,i+1}\cdot (f\dd h), g).
$$
\end{enumerate}
\end{construction}

\begin{example}
\label{ex:E-pos}
In case of join semilattices (where $h$ is the embedding of 
Example~\ref{exs:Pos}\eqref{item:semilattices}) 
the even step from $X_i$ to $X_{i+1}$ adds to every pair $x$, $y$ of 
elements of $X_i$ an upper bound compatible only with all elements under 
$x$ or $y$. And the odd step from $X_{i+1}$ to $X_{i+2}$ is a quotient
that turns this upper bound into a join of $x$ and $y$. 
After $\omega$ steps we get the join-semilattice reflection of $X$.
\end{example}

\begin{lemma}
\label{lem:l-chain}
Given a morphism $p_0:X_0\to P$ where $P$ 
is Kan-injective, there exists a unique 
cocone $p_i:X_i\to P$ ($i\in \Ord$) such that 
for all spans~\eqref{eq:4.1} the following triangle
\begin{equation}
\label{eq:4.4bis}
\vcenter{
\xymatrix{
A'
\ar[d]_{f\dd h}
\ar[dr]^{\Lan{h}{(p_i f)}}
&
\\
X_{i+1}
\ar[r]_{p_{i+1}}
&
P
}
}
\end{equation}
commutes.
\end{lemma}
\begin{proof}
We only need to prove the isolated step: given $p_i$ for 
$i$ even, we have unique $p_{i+1}$ and  $p_{i+2}$. For $p_{i+1}$
we observe that the morphisms $p_i:X_i\to P$ and $\Lan{h}{(p_i f)}:A'\to P$ 
form a cocone of the diagram defining 
$\xymatrix@1{X_{i}\ar@{-->}[r]&X_{i+1}}$. 
Indeed, the square 
$$
\xymatrix{
A
\ar[d]_f
\ar[r]^h
&
A'\ar[d]^{\Lan{h}{(p_if)}}
\\
X_i
\ar[r]_{p_i}
&
P
}
$$
clearly commutes.
It follows that there 
is a unique $p_{i+1}$ for which the above triangle commutes and 
which prolongs the given cocone.

Next we prove the existence of $p_{i+2}$ (uniqueness is clear since 
$\xymatrix@1{X_{i+1}\ar@{-->}[r]&X_{i+2}}$ is epic) by verifying 
that $p_{i+1}$ has the universal property of 
$\xymatrix@1{X_{i+1}\ar@{-->}[r]&X_{i+2}}$: for every 
square~\eqref{eq:4.3} we have
$$
\xymatrix{
A'
\ar[d]_{f\dd h}
\ar[r]^g
&
X_{i+1}
\ar[d]^{p_{i+1}}
\\
X_{j+1}
\ar[r]_{p_{j+1}}
\ar@{}[-1,1]|{\leq}
&
P
}
$$
Indeed, by~\eqref{eq:4.4bis}, the lower passage is $\Lan{h}{(p_j\cdot f)}$, 
hence, it is sufficient to verify 
$p_j\cdot f\leq p_{i+1}\cdot g\cdot h$. 
To that end, compose the given inequality~\eqref{eq:4.3} 
with $p_{i+1}$.
\end{proof}

\begin{remark}
\label{rem:r-chain}
In the Kan-injective reflection chain, for every pair
$i$, $j$ of even ordinals  
with $j\leq i$ and every span as in~\eqref{eq:4.1}
with $j$ in place of $i$, 
the connecting map $x_{i+1,i+2}$ merges the morphisms      
$(x_{ji} f)\dd h$ and $x_{j+1,i+1}\cdot (f\dd h)$.

Indeed, the equality~\eqref{eq:4.2} for $f$ implies
clearly the equality
$$
((x_{ji}f)\dd h) \cdot h
= 
x_{j+1,i+1}\cdot (f\dd h) \cdot h
$$
decomposes into two inequalities which by the universal property 
of the morphism $x_{i+1, i+2}$ gives rise to
$$
\begin{array}{lr}
x_{i+1,i+2}\cdot x_{j+1,i+1}\cdot f\dd h\leq x_{i+1,i+2}\cdot (x_{ji}f)\dd h  
&
\mbox{(putting $g=(x_{ji}f)\dd h$ in~\eqref{eq:4.3})},
\end{array}
$$
and
$$
\begin{array}{ll}
x_{i+1,i+2}\cdot (x_{j,i}f)\dd h\leq  x_{i+1,i+2}\cdot x_{j+1,i+1}\cdot f\dd h  
&
\mbox{(putting $g=x_{j+1,i+1}\cdot f\dd h$ in~\eqref{eq:4.3})}.
\end{array}
$$  
\end{remark}

\begin{theorem}
\label{th:t-chain}
If the Kan-injective reflection chain converges at an 
even ordinal $k$ (i.e., $x_{k,k+2}$ is invertible), 
then $X_k$ lies in $\LInj{\H}$ and 
$x_{0k}:X_0\to X_k$ is a reflection of $X_0$ in $\LInj{\H}$.  
\end{theorem}  
\begin{proof}
\mbox{}\hfill
\begin{enumerate}
\item 
We prove the Kan-injectivity of $X_k$. 
Given $h:A\to A'$ in $\X$ and $f:A\to X_k$, 
the square~\eqref{eq:4.2} allows us to define a morphism      
\begin{equation}
\label{eq:f/h}
\Lan{h}{f}
=
x_{k,k+2}^{-1}\cdot x_{k+1,k+2}\cdot (f\dd h):A'\to X_k
\end{equation}
and we verify the two properties needed. The first one 
is clear by applying~\eqref{eq:4.2} to $i=k$:
\begin{eqnarray*}
(\Lan{h}{f})\cdot h
&=&
x_{k,k+2}^{-1}\cdot x_{k+1,k+2}\cdot (f\dd h) \cdot h
\\
&=&
x_{k,k+2}^{-1}\cdot x_{k+1,k+2}\cdot x_{k,k+1}\cdot f
\\
&=&
x_{k,k+2}^{-1}\cdot x_{k,k+2}\cdot f
\\
&=&
f.
\end{eqnarray*}
For the second one let $g:A'\to X_k$ fulfil $gh\geq f$. 
Then we prove $g\geq \Lan{h}{f}$. 
The morphism $\ol{g}=x_{k,k+1}\cdot g$ fulfils 
$\ol{g}h\geq x_{k,k+1}\cdot f$, thus, the universal property 
of $x_{k+1,k+2}$ implies
$$
x_{k+1,k+2}\cdot \bar{g}\geq x_{k+1,k+2}\cdot (f\dd h).
$$
That is,
$$
x_{k,k+2}\cdot g
\geq 
x_{k+1,k+2}\cdot (f\dd h).
$$
By composing with $x_{k,k+2}^{-1}$ we get 
$g\geq x_{k,k+2}^{-1}\cdot x_{k+1,k+2}\cdot (f\dd h)$, as desired.
\item
\label{item:b}
Given $p:X_0\to P$ where $P$ lies in $\LInj{\H}$, 
we prove that the morphism $p_k$ of Lemma~\ref{lem:l-chain} 
belongs to $\LInj{\H}$.  
For every span~\eqref{eq:4.1} we want to prove that the bottom
triangle in the following diagram
$$
\xymatrix{ 
A
\ar[rr]^h
\ar[d]_f
&
&
A'
\ar[dll]^{\Lan{h}{f}}
\ar[d]^{\Lan{h}{(p_k f)}}
\\
X_k
\ar[rr]_{p_k}
&
&
P
}
$$
is commutative. Indeed,
$$
\begin{array}{rll}
p_k\cdot (f/h)
&
=
p_k\cdot x_{k,k+2}^{-1}\cdot x_{k+1,k+2}\cdot (f\dd h),
&
\mbox{by~\eqref{eq:f/h}}
\\
&
=(p_{k+2}\cdot x_{k,k+2})\cdot x_{k,k+2}^{-1}\cdot x_{k+1,k+2}\cdot (f\dd h)
&
\mbox{by Lemma~\ref{lem:l-chain}}
\\
&
=p_{k+2}\cdot x_{k+1,k+2}\cdot (f\dd h)
&
\\
&
=p_{k+1}\cdot (f\dd h),
&
\mbox{by Lemma~\ref{lem:l-chain}}
\\
&
=\Lan{h}{(p_{k}\cdot f)}
&
\mbox{again by Lemma~\ref{lem:l-chain}}
\end{array}
$$
\item
We have, for every $p$ as in~\eqref{item:b}, the morphism
$p_k$ of $\LInj{\H}$ with $p=p_k\cdot x_{0,k}$. Now
we prove the unicity of $p_k$.
It suffices to show that, given morphisms      
$b,b_0:X_k\to P$ with $b_0$ in $\LInj{\H}$, then
$$
b_0\cdot x_{0k}\leq b\cdot x_{0k}
\qquad 
\mbox{ implies }
\qquad b_0\leq b.
$$
Indeed, in the case where $b$ is also a morphism       
of $\LInj{\H}$ then the equality  
$b_0\cdot x_{0k}= b\cdot x_{0k}$ will imply $b_0=b$.
We are going to verify the above implication by proving that
$$
b_0\cdot x_{0k}\leq b\cdot x_{0k}
\qquad 
\mbox{ implies }
\qquad 
b_0\cdot x_{ik}\leq b\cdot x_{ik}
$$
for all $i\leq k$. We use transfinite induction. The first step $i=0$ 
is clear. Also limit steps are clear since the colimit cocones  
are collectivelly order-epic. 

It remains to check the isolated steps 
$i+1$ and $i+2$ for $i$ an even ordinal. 
\begin{enumerate}
\item 
From $i$ to $i+1$.
$$
\xymatrix{
A 
\ar[d]_f 
\ar[rr]^h
&
&
A'
\ar[d]^{f\dd h}
\\
X_i
\ar@{-->}[rr]
\ar@{-->}[dr]
&
&
X_{i+1}\ar@{-->}[dl]
\\
&
X_k
\ar@<-1ex>[rr]_-{b_0}
\ar@<1ex>[rr]^-{b}
&
&
P
}
$$
Since $x_{i,i+1}$ and all $f\dd h$ are collectively order-epic, 
we only need proving 
$$
b_0\cdot x_{i+1,k}\cdot f\dd h    
\leq 
b\cdot x_{i+1,k}\cdot f\dd h
$$

The formula~\eqref{eq:f/h} for $x_{ik}f$ in place of $f$ yields
$$
\Lan{h}{(x_{ik}f)}
=
x_{k,k+2}^{-1}\cdot x_{k+1,k+2}\cdot (x_{ik}f)\dd h.
$$
And, since $x_{k+1,k+2}$ merges $(x_{ik}f)\dd h$ 
and $x_{i+1, k+1}\cdot f\dd h$, see Remark~\ref{rem:r-chain}, we get
\begin{eqnarray*}
\Lan{h}{(x_{ik}f)}
&=&
x_{k,k+2}^{-1}\cdot x_{k+1,k+2}\cdot x_{i+1, k+1}\cdot f\dd h
\\
&=&
x_{k,k+2}^{-1}\cdot x_{k,k+2}\cdot x_{i+1, k}\cdot f\dd h
\\
&=&
x_{i+1, k}\cdot f\dd h.
\end{eqnarray*}
Since $b_0$ lies in $\LInj{\H}$, we know that 
$b_0[\Lan{h}{(x_{ik}f)}]=\Lan{h}{(b_0x_{ik}f)}$. And,
since by induction hypothesis $b_0x_{ik}\leq bx_{ik}$, 
we then obtain that $\Lan{h}{(b_0x_{ik})}\leq \Lan{h}{(bx_{ik})}$. 
Consequently:
\begin{eqnarray*}
b_0\cdot x_{i+1, k}\cdot f\dd h
&=&
b_0\cdot [\Lan{h}{(x_{ik}f)}]
\\
&=&
\Lan{h}{(b_0x_{ik}f)}
\\
&\leq&
\Lan{h}{(bx_{ik}f)}
\\
&\leq& 
b\cdot\Lan{h}{(x_{ik}f)}
\\
&=&
b\cdot x_{i+1, k}\cdot f\dd h
\end{eqnarray*}
\item
From $i+1$ to $i+2$. This is trivial because $x_{i+1,i+2}$ is order-epic.
\end{enumerate}
\end{enumerate}
\end{proof}

\begin{remark}
\label{rem:R-class}
The construction above can also be performed, 
assuming the base category $\X$ is cowellpowered, 
with every {\em class\/} $\H$ of morphisms, provided that it 
has the form $\H=\H_0\cup \H_e$ where $\H_0$ is small and 
$\H_e$ is a class of epimorphisms.
 
Indeed, in the isolated step $i\mapsto i+1$ with $i$ even 
the conical colimit exists because $x_{i,i+1}$ is the wide 
pushout of all the morphisms $\ol{h}$. If $h$ lies 
in $\H_e$ then $\ol{h}$ is an epimorphism. Thus cowellpoweredness 
guarantees that $X_{i+1}$ is obtained as a small wide pushout. 
The isolated step $i+1\mapsto i+2$ with $i$ even also makes no problem 
because $x_{i+1,i+2}$ is an epimorphism, and we obtain it as 
the cointersection of the corresponding epimorphisms over all subsets 
of $\H$.  
\end{remark}

\section{Locally ranked categories}
\label{sec:locally-ranked-cats}

Our main result, proved in Theorem~\ref{th:really-main} below, 
states that for every class $\H$ of morphisms      
in an order-enriched category $\X$ such that all but a set
of members of $\H$ are order-epic, the subcategory      
$\LInj{\H}$ is KZ-reflective. For that we need to assume 
that $\X$ is locally ranked, a concept introduced in~\cite{AHRT}. 
It is based on a factorization system $(\E,\M)$ in a (non-enriched) 
category $\X$ which is {\em proper\/}, i.e., all morphisms      
in $\E$ are epimorphisms and all morphisms in $\M$ are monomorphisms.  
An object $X$ of $\X$ has {\em rank $\lambda$\/}, where $\lambda$ is 
an infinite regular cardinal, provided that its  
hom-functor preserves unions of $\lambda$-chains of subobjects in $\M$.

\begin{definition}[See~\cite{AHRT}]
\label{def:D-rank}
An ordinary category $\X$ with a proper factorization system 
$(\E,\M)$ is called {\em locally ranked\/} if it is cocomplete and 
$\E$-cowellpowered, and every object has a rank.  
\end{definition}

\begin{remark}
\label{rem:R-EM}
In order-enriched categories {\em proper\/} is defined 
for a factorization system $(\E,\M)$ to mean that all morphisms      
in $\E$ are epimorphisms, and all morphisms in $\M$ are 
order-monomorphisms.  
\end{remark}

\begin{example}
\label{ex:E-EM}
Recall from~\cite{AHS} that every cocomplete, cowellpowered category      
has the factorization system $(\mathit{Epi},\mathit{Strong\ Mono})$. 
In every order-enriched category this factorization system is proper. 
Indeed, consider the inequality $mu\leq mv$ with $m$ a strong monomorphism,      
and let $c$ be the coinserter of $u$ and $v$.
$$
\xymatrix{
X
\ar@<0.8ex>[r]^v
\ar@<-0.8ex>[r]_u
&
A
\ar[d]_c
\ar[r]^m
&
B
\\
&
C
\ar@{.>}[ur]_{m'}}
$$
Then $m$ factorizes through $c$. But $c$ is an epimorphism       
and $m$ a strong monomorphism, thus $c$ is invertible. 
Equivalently, $u\leq v$.   
\end{example}

\begin{definition}
\label{def:D-EM}
Let $\X$ be an order-enriched category      
with a proper factorization system $(\E,\M)$. 
We call $\X$ {\em locally ranked\/} if it has weighted  
colimits, is $\E$-cowellpowered, and every object has a rank.   
\end{definition}

\begin{remark}
Explicitly, an object $A$ has rank $\lambda$  
iff given a union $X=\bigcup_{i<\lambda} m_i$ 
of a $\lambda$-chain $m_i:M_i\to X$ of subobjects  
in $\M$, then every morphism $p:A\to X$ factorizes through some $m_i$.

This concept is ``automatically enriched'': 
given $p,\, q: A \to X$ with $p\leq q$, it follows that there exists 
$i$ such that they both factorize through $m_i$:
$$
\xymatrixcolsep{1.5pc}
\xymatrixrowsep{1pc}
\xymatrix{
&
&
M_i
\ar[dd]^{m_i}
\\ 
\\
A
\ar@<-0.5ex>[rruu]_(.6)p
\ar@<1.1ex>[rruu]^q  
\ar@<-1.3ex>[rr]_{p'}
\ar@<0.3ex>[rr]^{q'}
&
&
X
}
$$
and we get $p'\leq q'$ from $m_i$ being an order-monomorphism.    

In other words: if the hom-functor into $\Set$ preserves 
$\lambda$-unions of $\M$-subobjects,  
it follows that the hom-functor into $\Pos$ also does.
\end{remark}

\begin{example}
\label{ex:E-Pos-lr}
\mbox{}\hfill
\begin{enumerate}
\item
$\Pos$ is a locally ranked category w.r.t. 
$(\mathit{Epi},\mathit{Strong\ Mono})$. Indeed, in the non-enriched 
sense all locally presentable categories      
are locally ranked, see~\cite{AHRT}, and, by Example~\ref{ex:E-EM},  
$(\mathit{Epi},\mathit{Strong\ Mono})$ is proper. 
From Examples~\ref{exs:coinserters}, \ref{ex:coproducts}
and Lemma~\ref{lem:coproducts+coinserters=colimits} we know
that $\Pos$ has weighted  colimits.
\item
$\Top_0$ is a locally ranked category w.r.t.  
$(\mathit{Surjection},\mathit{Subspace\ Embedding})$. 
Indeed, every space $A$ of cardinality less than $\lambda$ has 
rank $\lambda$ --- this  follows from unions of subspace embeddings in 
$\Top_0$ being carried by their unions in $\Set$. Cowellpoweredness w.r.t. 
surjective morphisms is obvious.
From Examples~\ref{exs:coinserters}, \ref{ex:coproducts}
and Lemma~\ref{lem:coproducts+coinserters=colimits} we know
that $\Top_0$ has weighted  colimits.
\end{enumerate}
\end{example}

\begin{remark}
\label{rem:R-Reit}
In Theorem~\ref{th:T-main} below  we use the following trick of Jan Reiterman, 
see~\cite{R} or~\cite{KR}. Given a transfinite chain 
$X:\Ord\to \X$ and an ordinal $i$, factorize all connecting maps 
$$
\xymatrix{
X_i
\ar[r]^{x_{ij}}
\ar@{->>}[d]_{e_{ij}}
&
X_j
\\
E_{ij}
\ar@{>->}[ru]_{m_{ij}}}
$$
in the $(\E,\M)$ factorization system. Since $\X$ is 
$\E$-cowellpowered there exists an ordinal $i^*$ such that 
all $e_{ij}$ with $j\geq i^*$ represent the same quotient of $X_i$. 
Define $\phi:\Ord\to \Ord$ by $\phi(0)=0$, $\phi(i+1)=\phi(i)^*$ 
and $\phi(i)=\bigvee_{j<i}\phi(j)$ for limit ordinals $i$. 
This gives a new transfinite chain 
$$
Y_i=E_{i,\phi(i)}
$$
and natural transformations 
$\beta_i=m_{i,\phi(i+1)}$ and $\gamma_i=e_{i,\phi(i+1)}$ 
with the following properties that were explicitly formulated
by Max Kelly~\cite{kelly:long}, Proposition~4.1. 
\end{remark}

\begin{lemma}
\label{lem:L-Reit}
For every transfinite chain $X:\Ord\to \X$ there exists 
a monotone function $\phi:\Ord\to \Ord$ preserving joins, 
a transfinite chain $Y:\Ord\to \X$ of $\M$-monomorphisms
and natural transformations 
$\gamma_i:X_i\to Y_i$ and $\beta_i:Y_i\to X_{\hat{i}}$, 
where $\hat{i}=\phi(i+1)$, such that
\begin{enumerate} 
\item 
$\beta_i\cdot \gamma_i=x_{i\hat{i}}$ for all $i\in \Ord$.
\item 
For all $j\geq \hat{i}$ we have a morphism of $\M$  
$$
\xymatrix{
Y_i
\ar[r]^{\beta_i}
&
X_{\hat{i}}
\ar[r]^{x_{ij}}
&
X_j
}
$$
\end{enumerate} 
and
\begin{enumerate} 
\setcounter{enumi}{2}
\item 
\label{item:reitermann3}
For every limit ordinal $j$ the union of the chain $Y_i$ ($i\leq j$) 
is given by 
$$
\xymatrix{
Y_i
\ar[r]^{\beta_i}
&
X_{\hat{i}}
\ar[r]^{x_{i\phi(j)}}
&
X_{\phi(j)}}
$$
\end{enumerate} 
\end{lemma}

\begin{remark}
Without loss of generality we choose $\phi$ so that 
$\hat{i}$ is an even ordinal for every ordinal $i$.  
\end{remark}

\begin{theorem}
\label{th:T-main}
For every set $\H$ of morphisms of a locally ranked category,  
$\LInj{\H}$ is a KZ-monadic subcategory.      
\end{theorem}
\begin{proof}
Since $\H$ is a set, there exists a cardinal $\lambda$ such that 
for every $h:A\to A'$ in $\H$ both $A$ and $A'$ have rank $\lambda$. 
Put 
$$
k=\phi(\lambda).
$$ 
We show that the connecting map 
$\xymatrix@1{X_0\ar@{-->}[0,1]&X_k}$
of the Kan-injective reflection chain, see Construction~\ref{cons:reflection},
is a reflection of $X=X_0$ in $\LInj{\H}$. 
\begin{enumerate}
\item 
\label{item:main1}
$X_{k}$ belongs to $\LInj{\H}$. Indeed,
given $h:A\to A'$ in $\H$ and $f:A\to X_{k}$, 
since $A$ has rank $\lambda$, there is some $i<\lambda$ 
making the diagram
$$
\xymatrix{
&
&
A
\ar[d]^f
\ar@/_1.3pc/@<0.3ex>[dll]_{f'}
\\
Y_i
\ar@/_2pc/[rr]_{m_i}
\ar[r]^{\beta_i}
&
X_{\hat{i}}\ar@{-->}[r]
&
X_{k}
}
$$
commutative. And we may choose this $i$ to be even.
Put 
\begin{equation}
\label{eq:6.0}
\Lan{h}{f}
=
x_{\hat{i}+1,k}\cdot (\beta_i f')\dd h  
\end{equation}
We show that it is the desired $\Lan{h}{f}$.
$$
\xymatrix{
A
\ar[0,1]^{h}
\ar[1,0]_{f'}
&
A'
\ar[2,0]^{(\beta_i f')\dd h}
\ar @/^1.5pc/ [2,1]^{\Lan{h}{f}}
&
\\
Y_i
\ar[1,0]_{\beta_i}
&
&
\\
X_{\hat{i}}
\ar @{-->}[0,1]
&
X_{\hat{i}+1}
\ar@{-->}[0,1]
&
X_k
}
$$
\begin{enumerate}
\renewcommand{\theenumii}{\theenumi\alph{enumii}}
\item 
$
(\Lan{h}{f})\cdot h
=
x_{\hat{i}+1,k}\cdot (\beta_i f')\dd h\cdot h
=
x_{\hat{i}+1,k}\cdot x_{\hat{i},\hat{i}+1}\cdot \beta_i\cdot f'
=
x_{\hat{i},k}\cdot  \beta_i\cdot f'
=
f
$.
\item
Let $g:A'\to X_{k}$ fulfil the inequality $f\leq gh$. We show that
$\Lan{h}{f} \leq g$. 

Again, the rank $\lambda$ of $A'$ ensures a factorization of $g$
for some ordinal $j<\gamma$: 
$$
\xymatrix{
&
&
A'
\ar[d]^g
\ar@/_1.3pc/@<0.3ex>[dll]_{g'}
\\
Y_j
\ar@/_2pc/[rr]_{m_j}
\ar[r]^{\beta_j}
&
X_{\hat{j}}
\ar@{-->}[r]
&
X_{k}
}
$$
And we may choose this $j$ to be even and fulfill $j\geq i$. Then the inequality 
$f\leq gh$ yields $m_j\cdot y_{ij}\cdot f'\leq m_j\cdot g'\cdot h$, and, since $m_j$ 
is order-monic, $y_{ij}\cdot f'\leq g'\cdot h$. Consequently, composing with 
$x_{\hat{j},\hat{j}+1}\cdot\beta_j$, and using the naturality of $\beta$, 
we obtain 
$$
x_{\hat{i},\hat{j}+1}\cdot \beta_i\cdot f'
=
x_{\hat{j},\hat{j}+1}\cdot\beta_j\cdot y_{ij}\cdot f'
\leq 
x_{\hat{j},\hat{j}+1}\cdot\beta_j\cdot g'\cdot h
.
$$
This is an instance of the inequality~\eqref{eq:4.3}
with $\beta_i\cdot f'$ in place of $f$ and 
$x_{\hat{j},\hat{j}+1}\cdot\beta_j\cdot g'$ in place of $g$. 
Hence, taking into account the universal property of the morphism 
$\xymatrix@1{X_{\hat{j}+1}\ar@{-->}[r]&X_{\hat{j}+2}}$,  
we conclude that 
$$
x_{\hat{j}+1,\hat{j}+2}\cdot x_{\hat{i},\hat{j}+1}\cdot (\beta_i\cdot f')\dd h
\leq 
x_{\hat{j}+1,\hat{j}+2}\cdot x_{\hat{j},\hat{j}+1}\cdot\beta_j\cdot g'
$$
from which it follows that $\Lan{h}{f}\leq m_j\cdot g'=g$.
\end{enumerate}
\item
\label{item:main2}
Let $p:X_0\to P$ be a morphism with $P\in \LInj{\H}$. 
Then we know that $p$ gives rise to a cocone $p_i:X_i\to P$ 
of the chain $X:\Ord\to \X$ as in Lemma~\ref{lem:l-chain}. 
We show that the morphism $p_{k}:X_{k}\to P$  
belongs to $\LInj{\H}$, i.e., the bottom triangle in the 
following diagram
$$
\xymatrix{
A
\ar[r]^h
\ar[d]_f
&
A'
\ar[dl]_{\Lan{h}{f}}
\ar[d]^{\Lan{h}{(p_{k}f)}}
\\
X_{k}\ar[r]_{p_{k}} 
&
P
}
$$
is commutative.

Indeed, given 
$f=m_i\cdot f'$, as in~\eqref{item:main1} above, 
then, recalling from~\eqref{item:main1} that 
$\Lan{h}{f}=x_{\hat{i}+1,k}\cdot (\beta_i f')\dd h$, 
and applying Lemma~\ref{lem:l-chain}, we have that:
$$
p_{k}\cdot\Lan{h}{f}
=
p_{\hat{i}+1}\cdot (\beta_i f')\dd h
=
\Lan{h}{[p_{\hat{i}}\cdot (\beta_if')]}
=
\Lan{h}{(p_{k}\cdot   x_{\hat{i},k}\cdot \beta_i\cdot f')}
=
\Lan{h}{(p_{k}\cdot f)}.
$$
\item
\label{item:main3}
In order to conclude that $p_k$ is unique, let $q:X_k\to P$ 
be another morphism of $\LInj{\H}$ with $q\cdot x_{0k}=p$. 
We prove that $q=p_k$ by showing, by transfinite induction,  
that $q\cdot x_{ik}=p_k\cdot x_{ik}$ for all $i\leq k$.

For $i=0$, this is the assumption. For limit ordinals the inductive 
step is trivial, by the universal property of the colimit. So we prove 
the property for $i+1$ and $i+2$ with $i$ even.
\begin{enumerate}
\renewcommand{\theenumii}{\theenumi\alph{enumii}}
\item
From $i$ to $i+1$.
Since $x_{i,i+1}$ and all $f\dd h$ are collectively epic, 
we only need proving 
$$
p_k\cdot x_{i+1,k}\cdot f\dd h   
= 
q\cdot x_{i+1,k}\cdot f\dd h
$$
for all $h\in\H$ and all $f$. For that, we first prove the equalities
\begin{equation}
\label{eq:form}
\Lan{h}{(x_{ik}\cdot f)}
=
x_{i+1,k}\cdot f\dd h, 
\qquad 
i<k.
\end{equation}
From Lemma~\ref{lem:L-Reit} we have that 
$x_{ik}\cdot f=x_{\hat{i}k}\cdot (\beta_i\cdot \gamma_i\cdot f)$, 
that is, $x_{ik}f=m_i(\gamma_if)$. Then, by~\eqref{eq:6.0},
we know that 
\begin{equation}
\label{eq:form2}
\Lan{h}{(x_{ik}\cdot f)}
=
x_{\hat{i}+1,k}\cdot (\beta_i\cdot \gamma_i\cdot f)\dd h
=
x_{\hat{i}+1,k}\cdot (x_{i\hat{i}}\cdot f)\dd h.
\end{equation} 
By Remark~\ref{rem:r-chain}, the morphism $x_{\hat{i}+1,\hat{i}+2}$ 
merges $(x_{i,\hat{i}}\cdot f)\dd h$ and $x_{i+1,\hat{i}+1}\cdot f\dd h$. 
Thus, $x_{\hat{i}+1,k}\cdot (x_{i,\hat{i}}\cdot f)\dd h=x_{i+1,k}\cdot f\dd h$. 
That is, by~\eqref{eq:form2}, $(x_{ik}\cdot f)/h=x_{i+1,k}\cdot f\dd h$.

Now, due to the equality $p_k\cdot x_{ik}=q\cdot x_{ik}$, 
we have $\Lan{h}{(p_k\cdot x_{ik})}=\Lan{h}{(q\cdot x_{ik})}$, 
hence $\Lan{h}{p_k\cdot (x_{ik})}=q\cdot \Lan{h}{(x_{ik})}$, because both $p_k$ 
and $q$ belong to $\LInj{\H}$. Using~\eqref{eq:form}, we obtain then that 
$p_k\cdot x_{i+1,k}\cdot f\dd h=q\cdot x_{i+1,k}\cdot f\dd h$.
\item
From $i+1$ to $i+2$. This is clear, since $x_{i+1,i+2}$ is an order-epimorphism.
\end{enumerate}
\item
From~\eqref{item:main2} and~\eqref{item:main3} we know that $\LInj{\H}$
is reflective, therefore KZ-monadic by 
Corollary~\ref{cor:reflective=>KZ-monadic}.
\end{enumerate}
\end{proof}

\begin{theorem}
\label{th:really-main}
In every locally ranked, order-enriched category $\X$ 
the subcategory $\LInj{\H}$ is KZ-monadic for every class
$$
\H
=
\H_0\cup\H_e
$$
of morphisms with $\H_0$ small and $\H_e$ consisting
of order-epimorphisms.
\end{theorem}
\begin{proof}
\mbox{}\hfill
\begin{enumerate}
\item 
Since the members of $\H_e$ are order-epimorphisms, the category
$\LInj{\H_e}$ is simply the orthogonal (full) subcategory
$\H_e^\perp$, see Example~\ref{ex:orthogonal}. 
It was proved in~2.4(c) of~\cite{AHS2} that 
$\H_e^\perp$ is again a locally ranked category w.r.t.
$\E$ = all epis and $\M$ = all monics lying in $\H_e^\perp$.
(The proof concerned ordinary categories, but it adapts
immediately to the order-enriched setting.)

Moreover, $\H_e^\perp$ is a reflective subcategory of $\X$
whose units are order-epimorphisms. Indeed, the reflection
of an object $X$ of $\X$ is the wide pushout of all morphisms
$\ol{h}$ in all pushouts~\eqref{eq:5.2a}.

Since $h$ is an order-epimorphism and $\X$ has weighted
colimits (thus, $\ol{h}$ and $\ol{f}$ are collectively order-epic),
it is clear that $\ol{h}$ is also an order-epimorphism. Analogously,
a wide pushout of order-epimorphisms is an order-epimorphism.
Thus, if $R:\X\to\H_e^\perp$ denotes the reflector, the units
$\eta_X:X\to RX$ are all order-epimorphisms.
\item
The set 
$$
\wh{\H_0}
=
\{ Rh\mid \mbox{$h$ in $\H_0$}\}
$$
of morphisms of the locally ranked category $\H_e^\perp$
fulfills, by  Theorem~\ref{th:T-main}, that
$$
{\mathsf{LInj}}_{\H_{e}^{\perp}}(\wh{\H_0})
\mbox{ is reflective in $\H_e^\perp$.}
$$
(The lower index is used to stress in which category the
injectivity is considered.) Consequently, 
${\mathsf{LInj}}_{\H_{e}^{\perp}}(\wh{\H_0})$
is a reflective subcategory of $\X$. The theorem will
be proved by verifying that
$$
{\mathsf{LInj}}_{\X}(\H)
=
{\mathsf{LInj}}_{\H_{e}^{\perp}}(\wh{\H_0}).
$$
We prove that~(a) ${\mathsf{LInj}}_{\X}(\H)$ is
a subcategory of ${\mathsf{LInj}}_{\H_{e}^{\perp}}(\wh{\H_0})$
and~(b) the other way round.
\begin{enumerate}
\item[(a1)] 
Every object $X$ of $\X$ Kan-injective w.r.t. $\H$
is clearly an object of $\H_e^\perp$; we prove that it is
Kan-injective w.r.t. $Rh$ in $\wh{\H_0}$.
$$
\xymatrixcolsep{.1pc}
\xymatrix{
A
\ar[0,4]^-{h}
\ar[1,1]^{\eta_A}
\ar@/_1pc/ [2,2]_{f\eta_A}
&
&
&
&
A'
\ar[1,-1]_{\eta_{A'}}
\ar@/^1pc/ [2,-2]^{\Lan{h}{(f\eta_A)}}
\\
&
RA
\ar[0,2]^-{Rh}
\ar[1,1]^{f}
&
&
RA'
\ar[1,-1]_{\wh{f}}
&
\\
&
&
X
&
&
}
$$
Given $f:RA\to X$, the morphism $\Lan{h}{(f\eta_A)}$
factorises, since $X$ is in $\H_e^\perp$, through $\eta_{A'}$:
we have a unique $\wh{f}$ such that the diagram above commutes.
Then
$$
\wh{f}
=
\Lan{Rh}{f}.
$$
Indeed, $\wh{f}\cdot Rh=f$. And given $g:RA'\to X$ with
$f\leq g\cdot Rh$, then $f\cdot\eta_A\leq g\cdot Rh\cdot\eta_A=g\cdot\eta_{A'}\cdot h$
which implies $\Lan{h}{(f\eta_A)}\leq g\cdot\eta_{A'}$. Recall that
$R$ is a reflector of $\H_e^\perp$ and $\eta_{A'}$ is an order-epimorphism.
Thus $\wh{f}\leq g$, as desired.
\item[(a2)]
Every morphism $p:X\to Y$ of $\X$ Kan-injective w.r.t. $\H$ lies
in the (full) subcategory $\H_e^\perp$, and we must prove that
$p$ is Kan-injective w.r.t. $Rh$. Given $f:RA\to X$ we have
seen that $\wh{f}=\Lan{Rh}{f}$ above, and analogously for 
$f_1=p\cdot f:RA\to Y$ we have $\wh{f_1}$, defined by 
$\wh{f_1}\cdot\eta_{A'}=\Lan{h}{(f_1 \eta_A)}$, satisfying
$\wh{f_1}=\Lan{Rh}{f_1}$. Since $p$ is Kan-injective w.r.t.
$\H$, we have 
$$
p\cdot \wh{f}\cdot\eta_{A'}
=
p\cdot\Lan{h}{(f\eta_A)}
=
\Lan{h}{(pf\eta_A)}
=
\Lan{h}{(f_1\eta_A)}
=
\wh{f_1}\cdot\eta_{A'}
$$
and this implies $p\cdot\wh{f}=\wh{f_1}$ since $\eta_{A'}$
is order-epic. Thus
$$
p\cdot (\Lan{Rh}{f})
=
p\cdot\wh{f}
=
\wh{f_1}
=
\Lan{Rh}{(pf)}
$$
as required.
\item[(b1)]
Every object $X$ of $\H_e^\perp$ Kan-injective w.r.t. 
$\wh{\H_0}$ is Kan-injective w.r.t. $\H$. We only need to
consider $h:A\to A'$ in $\H_0$.
$$
\xymatrixcolsep{1pc}
\xymatrix{
A
\ar[0,4]^-{h}
\ar[1,1]^{\eta_A}
\ar@/_2pc/ [3,2]_{f}
&
&
&
&
A'
\ar[1,-1]_{\eta_{A'}}
\ar@/^2pc/ [3,-2]^{\Lan{h}{f}}
\\
&
RA
\ar[0,2]^-{Rh}
\ar[2,1]_(.4){f^\sharp}
&
&
RA'
\ar[2,-1]^(.4){\Lan{Rh}{f^\sharp}}
&
\\
&
&
&
&
\\
&
&
X
&
&
}
$$
Given $f:A\to X$, since $X$ is in $\H_e^\perp$, we have 
a unique $f^\sharp:RA\to X$ with $f=f^\sharp\eta_A$. 
And we define
$$
\Lan{h}{f}
=
(\Lan{Rh}{f^\sharp})\cdot\eta_{A'}.
$$
This morphism has both of the required properties: firstly
\begin{eqnarray*}
(\Lan{h}{f})\cdot h
&=&
(\Lan{Rh}{f^\sharp})\cdot\eta_{A'}\cdot h
\\
&=&
(\Lan{Rh}{f^\sharp})\cdot Rh\cdot\eta_A
\\
&=&
f^\sharp\cdot\eta_A
\\
&=&
f.  
\end{eqnarray*}
Secondly, given $g:A'\to X$ with $f\leq g\cdot h$, there exists
a unique $g^\sharp:RA'\to X$ with $g=g^\sharp\cdot\eta_{A'}$.
From
$$
f^\sharp\cdot\eta_A
=
f
\leq
g\cdot h
=
g^\sharp\cdot\eta_{A'}\cdot h
=
g^\sharp\cdot Rh\cdot\eta_A
$$
we derive, since $\eta_A$ is an order-epimorphism, that
$f^\sharp\leq g^\sharp\cdot Rh$. Since clearly
$\Lan{Rh}{(g^\sharp Rh)}\leq g^\sharp$, we conclude
\begin{eqnarray*}
\Lan{h}{f}
&=&
(\Lan{Rh}{f^\sharp})\cdot\eta_{A'}
\\
&\leq&
\left(\Lan{Rh}{(g^\sharp Rh)}\right)\cdot\eta_{A'}  
\\
&\leq&
g^\sharp\cdot\eta_{A'}
\\
&=&
g.
\end{eqnarray*}
\item[(b2)]
Every morphism $p:X\to Y$ of $\H_e^\perp$ Kan-injective
w.r.t. $\H_0$ is Kan-injective w.r.t. $\H$. Again, we only
need to consider $h$ in $\H_0$. Given $f:A\to X$ we have
$\Lan{h}{f}=(\Lan{Rh}{f^\sharp})\cdot\eta_{A'}$. Put $f_1=p\cdot f$
and obtain the corresponding $f_1^\sharp:RA\to Y$ with 
$\Lan{h}{f_1}=(\Lan{Rh}{f_1^\sharp})\cdot\eta_{A'}$. Then 
$f_1=p\cdot f$ implies $f_1^\sharp\cdot\eta_A=p\cdot f^\sharp\cdot\eta_A$,
and since $\eta_A$ is an order-epimorphism, we conclude
$f_1^\sharp=p\cdot f^\sharp$. Consequently, from the Kan-injectivity
of $p$ w.r.t. $Rh$ we obtain the desired equality:
\begin{eqnarray*}
p\cdot (\Lan{h}{f})
&=&
p\cdot (\Lan{Rh}{f^\sharp})\cdot\eta_{A'}
\\
&=&
\left(\Lan{Rh}{(pf^\sharp)}\right)\cdot\eta_{A'}
\\
&=&
(\Lan{Rh}{f_1^\sharp})\cdot\eta_{A'}
\\
&=&
\Lan{h}{f_1}
\\
&=&
\Lan{h}{(pf)}.  
\end{eqnarray*}
\end{enumerate}
\end{enumerate}
\end{proof}

\section{A counterexample}
\label{sec:counterexample}

We give an example of a proper class $\H$ of continuous maps in
$\Top_0$ for which the Kan-injectivity category $\LInj{\H}$
is not reflective. The example is based on ideas of~\cite{ar88}.
\begin{enumerate}
\item
We denote by $\C$ the following category
$$
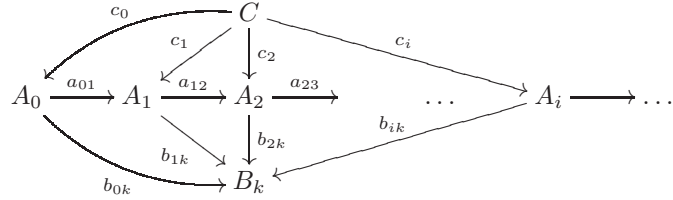
\xymatrixrowsep{1.5pc}
\xymatrix{
&
&
C
\ar@/_1pc/[1,-2]_{c_{0}}
\ar[1,-1]_{c_{1}}
\ar[1,0]^{c_{2}}
\ar[1,3]^{c_i}
&
&
&
&
\\
A_0
\ar[0,1]^-{a_{01}}
\ar@/_1pc/[1,2]_{b_{0k}}
&
A_1
\ar[0,1]^-{a_{12}}
\ar[1,1]_{b_{1k}}
&
A_2
\ar[0,1]^-{a_{23}}
\ar[1,0]^{b_{2k}}
&
&
\dots
&
A_i
\ar[0,1]
\ar[1,-3]_{b_{ik}}
&
\dots
\\
&
&
B_k
&
&
&
&
}
$$
It consists of a transfinite chain $a_{ij}:A_i\to A_j$
($i\leq j$ in $\Ord$) and, for every ordinal $k$, a cocone
$b_{ik}:A_i\to B_k$ ($i\in\Ord$) of that chain. Furthermore, there are 
morphisms $c_i:C\to A_i$ ($i$ in $\Ord$) with free composition
modulo the equations
$$
b_{kk}\cdot c_k = b_{ik}\cdot c_i,
\quad
\mbox{for all $i\geq k$}
$$
In particular, we have
$$
b_{kk}\cdot c_k\not= b_{ik}\cdot c_i,
\quad
\mbox{for all $i<k$}
$$

This category is concrete, i.e., it has a faithful functor
into $\Set$. For example, take $U:\C\to\Set$ with 
$UB_i=UA_i=\{t\in\Ord\mid t\leq i\}$ and $UC=\{0\}$.
The morphisms $Ua_{ij}$ are then the inclusions, 
$Ub_{ik}(t)=\max(t,k)$ and $Uc_i(0)=i$.

V\'{a}clav Koubek proved in~\cite{koubek} that every concrete category
has an {\em almost full embedding\/} $E:\C\to\Top_2$ into the category
$\Top_2$ of topological Hausdorff spaces. This means that $E$
is faithful and maps morphisms of $\C$ into nonconstant mappings,
and every nonconstant continous map $p:EX\to EY$ has the form
$p=Ef$ for a unique $f:X\to Y$ in $\C$.
\item
\label{item:counterexample2}
For the proper class
$$
\H=\{ Ea_{0i}\mid i\in\Ord\}
$$
in $\Top_0$ we prove that the space $EA_0$ does not have a reflection
in $\LInj{\H}$. We first verify that all spaces $EB_k$
are Kan-injective:
$$
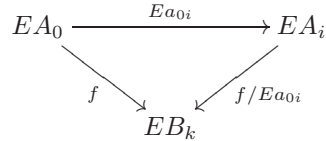
\xymatrix{
EA_0
\ar[0,2]^-{Ea_{0i}}
\ar[1,1]_{f}
&
&
EA_i
\ar[1,-1]^{\Lan{Ea_{0i}}{f}}
\\
&
EB_k
&
}
$$
Given $i\in\Ord$ and $f:EA_0\to EB_k$ we find $\Lan{Ea_{0i}}{f}$
as follows:
\begin{enumerate}
\item 
\label{item:counterexample2a}
If $f$ is nonconstant, then $f=Eb_{0k}$ and we claim that 
$\Lan{Ea_{0i}}{f}=Eb_{ik}$. For that it is sufficient to recall
that $EB_k$ is a Hausdorff space, thus, given $g:EA_i\to EB_k$
with $f\leq g\cdot Ea_{0i}$, it follows that $f=g\cdot Ea_{0i}$. Hence, 
$g$ is also nonconstant. But then $g=Eb_{ik}$.
\item
If $f$ is constant, then we claim that $\Lan{Ea_{0i}}{f}$ is the
constant function with the same value. For that, take again $g$
with $f\leq g\cdot Ea_{0i}$ and conclude $f=g\cdot Ea_{0i}$. This implies
that $g$ is constant (and thus $g=\Lan{Ea_{0i}}{f}$) because
otherwise $g=Eb_{ik}$, but the latter implies $f=Eb_{ik}\cdot Ea_{0i}=Ea_{ik}$
which is nonconstant --- a contradiction.
\end{enumerate}
\item
Suppose that $r:EA_0\to R$ is a reflection of $EA_0$ in $\LInj{\H}$.
We derive a contradiction by proving that there exists a proper class
of continuous functions from $EC$ to $R$.

Since $r$ is Kan-injective, for every $i\in\Ord$ we have
$$
r_i=\Lan{Ea_{0i}}{r}:EA_i\to R
$$
And the Kan-injectivity of $EB_k$ implies that there exists
a Kan-injective morphism
$$
s_k:R\to EB_k
\quad
\mbox{with $Eb_{0k}=s_k\cdot r$}
$$
See the diagram
$$
\xymatrixrowsep{1.5pc}
\xymatrix{
EA_0
\ar[0,2]^-{Ea_{0i}}
\ar[1,1]^{r}
\ar[2,1]_{Eb_{0k}}
&
&
EA_i
\ar[1,-1]_{r_i}
\ar[2,-1]^{Eb_{ik}}
\\
&
R
\ar[1,0]^(.3){s_k}
&
\\
&
EB_k
&
}
$$
Then, due to Kan-injectivity of $s_k$, we have
$$
s_k\cdot r_i
=
s_k\cdot (\Lan{Ea_{0i}}{r})
=
\Lan{(Eb_{0k})}{(Ea_{0i})}
$$
and in part~\eqref{item:counterexample2a} above
we have seen that the last morphism is $Eb_{ik}$.
Thus the above diagram commutes.
For all $k> i$ we have $b_{kk}\cdot c_k\not= b_{ik}\cdot c_i$,
therefore, $Eb_{kk}\cdot Ec_k\not= Eb_{ik}\cdot Ec_i$. Thus
$$
s_k\cdot r_k\cdot Ec_k
\not=
s_k\cdot r_i\cdot Ec_i
$$
which implies 
$$
r_k\cdot Ec_k
\not=
r_i\cdot Ec_i:EC\to R
$$
for all $k>i$ in $\Ord$. This is the desired contradiction.
\end{enumerate}

\section{Weak Kan-injectivity and right Kan-injectivity}
\label{sec:weak-kan-inj}

It may seem more natural to define left Kan-injectivity of an
object $X$ w.r.t. $h:A\to A'$ by requiring only that for every morphism
$f:A\to X$ a left Kan extension $\Lan{h}{f}:A'\to X$ exists. Thus,
we only have $f\leq (\Lan{h}{f})\cdot h$, but not necessarily an equality.

\begin{example}
\label{ex:5.2}
For the morphism
$$
\let\objectstyle=\scriptstyle
\xy <1 pt,0 pt>:
    (000,000)  *++={};
    (030,020) *++={} **\frm{.};
    (070,000)  *++={};
    (100,020) *++={} **\frm{.}
\POS(005,010) *{\bullet};
    (025,010) *{\bullet};    
    (085,010) *{\bullet};
\POS(050,012) *{\stackrel{h}{\to}}
\endxy
$$  
in $\Pos$, the left Kan-injective objects in the above
weak sense are precisely the join-semilattices.
\end{example}

\begin{definition}
Let $h:A\to A'$ be a morphism. 
\begin{enumerate}
\item 
An object $X$ is called {\em weakly left Kan-injective\/}
w.r.t. $h$ if for every morphism $f:A\to X$ a left Kan
extension $\Lan{h}{f}:A'\to X$ of $f$ along $h$ exists.
\item
A morphism $p:X\to Y$ between weakly left Kan-injective
objects is called weakly left Kan-injective if
$p\cdot (\Lan{h}{f})=\Lan{h}{(pf)}$ holds for all $f:A\to X$.
\end{enumerate}
\end{definition}

\begin{remark}
When comparing Examples~\ref{ex:5.2} and~\ref{exs:Pos} 
we see that in some cases (strong) left Kan-injectivity seems 
more ``natural'' than the weak one. Theorem~\ref{th:weak=strong}
indicates that the weak notion is, moreover, not really needed.
\end{remark}

\begin{notation}
For every class $\H$ of morphisms of an order-enriched
category $\X$ we denote by
$$
\wLInj{\H}
$$
the category of all objects and morphisms of $\X$ that are
weakly left Kan-injective w.r.t. all members of $\H$.  
\end{notation}

\begin{theorem}
\label{th:weak=strong}
In every locally ranked order-enriched category $\X$,
given a set $\H$ of morphisms there exists a class $\ol{\H}$
of morphisms such that  
$$
\wLInj{\H}=\LInj{\ol{\H}}
$$  
\end{theorem}
\begin{proof}
\mbox{}\hfill
\begin{enumerate}
\item
The category $\X$ has cocomma objects, i.e., given a span
$
\xymatrix@1{
A
&
D
\ar[0,1]^-{q}
\ar[0,-1]_-{p}
&
B
}
$ 
there exists a couniversal square
$$
\xymatrix{
D
\ar[0,1]^-{q}
\ar[1,0]_{p}
&
B
\ar[1,0]^{\ol{q}}
\\
A
\ar[0,1]_-{\ol{p}}
\ar@{}[-1,1]|{\leq}
&
C
}
$$
Its construction is analogous to the construction
of pushouts via coequalisers: form a coproduct
$
\xymatrix@1{
A
\ar[0,1]^-{i_A}
&
A+B
&
B
\ar[0,-1]_-{i_B}
}
$
and a coinserter
$$
\xymatrix{
D
\ar[0,1]^-{i_B\cdot q}
\ar[1,0]_{i_A\cdot p}
&
A+B
\ar[1,0]^{c}
\\
A+B
\ar[0,1]_-{c}
\ar@{}[-1,1]|{\leq}
&
C
}
$$ 
Then put $\ol{p}=c\cdot i_A$ and $\ol{q}=c\cdot i_B$.
\item
The category $\wLInj{\H}$ is reflective. The proof is completely
analogous to that of Theorem~\ref{th:t-chain}, except 
that Construction~\ref{cons:reflection} needs one modification: in 
diagram~\eqref{eq:4.2} we do not require equality but
inequality:
$$
\xymatrix{
A
\ar[d]_f
\ar[r]^h
&
A'
\ar[d]^{f\dd h}
\\
X_i
\ar@{-->}[r]
\ar@{}[-1,1]|{\leq}
&
X_{i+1}
}
$$
Thus, given $h$ in $\H$ and $f:A\to X_i$ we form a cocomma object
$$
\xymatrix{
A
\ar[d]_{f}
\ar[r]^-{h}
&
A'
\ar[d]^{\ol{f}}
\\
X_i
\ar[r]_-{\ol{h}}
\ar@{}[-1,1]|{\leq}
&
C
}
$$
Then
$
\xymatrix{
X_i
\ar@{-->}[0,1]
&
X_{i+1}
}
$
is the wide pushout of all $\ol{h}$ (with the colimit
cocone $c_{f,h}:C\to X_{i+1}$) and we put 
$f\dd h=c_{f,h}\cdot\ol{f}$.
\item
The category $\wLInj{\H}$ is also inserter-ideal: the proof
is completely analogous to that of Lemma~\ref{lem:inserter-ideal}.
By Theorem~\ref{th:reflective=>KZ-monadic} $\wLInj{\H}$
is a KZ-monadic category.
\item
Let $\ol{\H}$ denote the collection of all reflection maps
of objects of $\X$ in $\wLInj{\H}$. Then
$$
\wLInj{\H}
=
\LInj{\ol{\H}}
$$
holds by Proposition~\ref{prop:4.5}.
\end{enumerate}
\end{proof}

\begin{remark}
There is another obvious variation of Kan-injectivity, using right 
Kan extensions instead of left ones. Given $h:A\to A'$ and $f:A\to X$ 
we denote by $\Ran{h}{f}:A'\to X$ the largest morphism       
with 
$$
\xymatrixcolsep{1pc}
\xymatrix{
A
\ar[rr]^-h
\ar[dr]_f
&
\ar@{}[1,0]|(.4){\geq}
&
A'
\ar[dl]^{\Ran{h}{f}}
\\
&
X
&
}
$$
\end{remark}

\begin{definition}
\label{def:D-right}
\mbox{}\hfill
\begin{enumerate}
\item 
An object $X$ is {\em right Kan-injective\/} w.r.t. $h:A\to A'$ 
provided that for every morphism $f:A\to X$ a right Kan extension 
$\Ran{h}{f}$ exists and fulfils 
$$
f=(\Ran{h}{f})\cdot h.
$$
\item
A morphism $p:X\to Y$ is right Kan-injective
w.r.t. $h:A\to A'$ provided that both $X$ and $Y$ 
are, and for every morphism $f:A\to X$ we have 
$$
p\cdot (\Ran{h}{f})=\Ran{h}{(pf)}.
$$
\end{enumerate}
\end{definition}

\begin{notation}
\label{not:N-right}
$\RInj{\H}$  is the subcategory of all right Kan-injective objects and 
morphisms w.r.t. all members of $\H$.  
\end{notation}

\begin{remark}
\label{rem:R-right}
If $\X^\co$ denotes the category obtained from $\X$ by reversing 
the ordering of homsets 
(thus leaving objects, morphisms and composition as before), 
then every class $\H$ of morphisms in $\X$ yields a right 
Kan-injectivity subcategory $\RInj{\H}$ of $\X$ as well as 
a left Kan-injectivity subcategory $\LInj{\H}$ in $\X^\co$, 
and we have
$$\RInj{\H}=(\LInj{\H})^\co.
$$
Thus, in a sense, right Kan-injectivity is not needed. 
However, in some examples it is more intuitive to work with this concept.
\end{remark}

\begin{example}
\label{ex:E-right}
We have considered $\Top_0$ above as an ordered category with respect 
to the specialisation order. Thus $\Top_0^\co$ is the same
category with dual of the specialisation order on homsets.
This is the prefered enrichment of many authors. 
The examples of $\LInj{\H}$ in Section~\ref{sec:lkan-inj} 
become, under the last enrichment of $\Top_0$, examples of $\RInj{\H}$.
\end{example}

\section{Conclusion and open problems}

For locally ranked categories (which is a wide class containing
all locally presentable categories and $\Top$) it is known that
orthogonality w.r.t. a set of morphisms defines a full
reflective subcategory. And the latter is the Eilenberg-Moore
category of an idempotent monad. In our paper we have proved the
order-enriched analogy: given an order-enriched, locally ranked
category, then Kan-injectivity w.r.t. a set of morphisms defines
a (not generally full) reflective subcategory. The monad this creates
is a Kock-Z\"{o}berlein monad whose Eilenberg-Moore category is the
given subcategory. And conversely, every Eilenberg-Moore category of
a Kock-Z\"{o}berlein monad is specified by Kan-injectivity w.r.t.
all units of the monad. On the other hand, we have presented a class of
continuous maps in $\Top_0$ whose Kan-injectivity class is not
reflective.

Our main technical tool was the concept of an inserter-ideal
subcategory: we proved that every inserter-ideal reflective 
subcategory is the Eilenberg-Moore category of a Kock-Z\"{o}berlein
monad. And given any class of morphisms, Kan-injectivity always
defines an inserter-ideal subcategory.

It is easy to see that for every set of morphisms in a locally presentable
category the Kan-injectivity subcategory is accessibly embedded, i.e.,
closed under $\kappa$-filtered colimits for some infinite cardinal $\kappa$.
It is an open problem whether every inserter-ideal, accessibly embedded
subcategory closed under weighted limits is the Kan-injectivity subcategory
for some set of morphisms. This would generalise the known fact that the
orthogonality to sets of morphisms defines precisely the full, accessibly
embedded subcategories closed under limits, see~\cite{ar}.

In case of orthogonality, a morphism $h$ is called a consequence
of a set $\H$ of morphisms provided that objects orthogonal to $\H$
are also orthogonal w.r.t. $h$. A simple logic of orthogonality,
making it possible to derive all consequences of $\H$, is known~\cite{AHS2}.
Despite the strong similarity between orthogonality and 
Kan-injectivity, we have not been so far able to find a (sound and complete)
logic for Kan-injectivity.

\end{document}